\documentclass[11pt,a4paper]{article}

\usepackage{amsmath}
\usepackage{amsfonts}
\usepackage{amssymb}
\usepackage{amsthm,bbm}
\usepackage{a4wide}
\usepackage[]{amsmath,amssymb,amsfonts,latexsym,amsthm,enumerate,fullpage}
\usepackage[dvipsnames]{xcolor}
\usepackage{authblk}

\newcommand{\bN}{\mathbb{N}}
\newcommand{\bZ}{\mathbb{Z}}
\newcommand{\bQ}{\mathbb{Q}}
\newcommand{\bR}{\mathbb{R}}
\newcommand{\bC}{\mathbb{C}}
\newcommand{\bF}{\mathbb{F}}
\newcommand{\bP}{\mathbb{P}}
\newcommand{\bA}{\mathbb{A}}

\newcommand{\mB}{\mathcal{B}}
\newcommand{\mV}{\mathcal{V}}
\newcommand{\mX}{\mathcal{X}}
\newcommand{\mM}{\mathcal{M}}
\newcommand{\mO}{\mathcal{O}}

\newcommand{\p}{\mathfrak{p}}

\numberwithin{equation}{section}

\newcommand\sett[2]{\left\{ #1 \left| \; \vphantom{#1 #2} \right. #2  \right\}}
\newcommand{\set}[1]{\{#1\}}
\newcommand{\edge}[1]{[#1]}
\newcommand{\redge}[1]{[#1)}
\newcommand{\ledge}[1]{(#1]}
\newcommand{\Card}[1]{\left\lvert#1\right\rvert}

\def\mpar{M^{||}}

\def\npar{n^{\parallel}}
\def\na{n_1}
\def\m{n_2}

\newtheorem{thm}{Theorem}[section]
\newtheorem{pro}[thm]{Proposition}
\newtheorem{lem}[thm]{Lemma}
\newtheorem{cor}[thm]{Corollary}

\newtheorem{que}[thm]{Question}

\theoremstyle{Definition}
\newtheorem{dfn}[thm]{Definition}
\newtheorem{alg}[thm]{Algorithm}

\newtheorem{rem}[thm]{Remark}

\theoremstyle{plain}
\def\dist{{\rm dist}}
\def\gp{{G^\parallel}}

\begin{document}

\title{Good Locally Testable Codes}
\author[1]{Irit Dinur \thanks{irit.dinur@weizmann.ac.il}}
\author[2]{Shai Evra \thanks{shai.evra@mail.huji.ac.il}}
\author[2]{Ron Livne \thanks{ron.livne@mail.huji.ac.il}}
\author[1,2]{Alexander Lubotzky \thanks{alex.lubotzky@mail.huji.ac.il}}
\author[2]{Shahar Mozes \thanks{mozes@math.huji.ac.il}}
\affil[1]{Weizmann Institute, Rehovot, Israel}
\affil[2]{Hebrew University, Jerusalem, Israel}
\maketitle

\begin{abstract}
An explicit construction of locally testable codes of constant rate, constant distance and constant number of queries is given.
Hence answering affirmatively the $c^3$-problem.
\end{abstract}

\section{Introduction}

An $[n,k,d]$ binary linear error correcting code is a subspace $C$ of $\bF_2^n$, of dimension $k = \dim C$, and distance $d = \min\{ \mbox{wt}(c) \,:\, 0 \ne c \in C \}$, where $\mbox{wt}(c) = |\{ 1 \leq i \leq n \,|\, c_i \ne 0\}|$ is the Hamming weight. 
Call $\rho = \rho(C) = \frac{k}{n}$, the rate of $C$, and $\delta = \delta(C) = \frac{d}{n}$, the normalized distance of $C$.
A code (or more precisely, a family of codes where $n\rightarrow \infty$) is called good if there exists $\epsilon > 0$ such that $\rho$ and $\delta$ are bounded from below by $\epsilon$.

\begin{dfn}
The code $C$ is called an LTC (locally testable code) if it has a $(q,\kappa)$-tester for $q \in \bN$, $\kappa >0$ independent of $n$, where a  $(q,\kappa)$-tester $T$ is a probabilistic algorithm which given a word $f \in \bF_2^n$, queries $q$ bits from $f$ and outputs Accept or Reject, such that
\begin{itemize}
\item If $f \in C$, then $\mathbb{P}[T \mbox{ Accepts } f] = 1$.
\item If $f \not\in C$, then $\mathbb{P}[T \mbox{ Rejects } f] \geq \kappa \cdot \dist(f,C)$, where $\dist(f,C) = \frac{1}{n}\min\{ \mbox{wt}(f - c) \,:\, c \in C\}$.
\end{itemize}
\end{dfn}

A code is called LDPC (low density parity check) if it is defined by a set of constraints (namely, equations) involving a bounded number of coordinates. Locally testable codes are always LDPC but not vice versa.

In the late 1940's Hamming \cite{Ham} defined the notion of an error correcting code and proved that random linear subspaces are good codes with high probability. In the 1960's Gallager \cite{Gal} defined LDPC codes and showed that a random LDPC code is good with high probability. 
It then took approximately thirty years until an explicit (non-random) construction of good LDPC codes was given by the breakthrough work of Sipser and Spielman \cite{SS} in the 90's, building on the previous work of Tanner \cite{Tan}.

Locally testable codes were defined in the 90's (\cite{BFLS, Aro, RS, GS}) with motivation from the theory of probabilistically checkable proofs, yet the existence of good LTCs remained an outstanding problem, called the $c^3$-problem (constant rate, constant distance and constant number of queries).
Experts have gone back and forth on the question whether such codes can at all exist (see \cite[Conjecture~3.4]{Gol1} and \cite[Section3.3.2]{Gol2}, and see \cite{DELLM} for a detailed history).  
Part of the difficulty comes from the fact that random codes are {\em not} LTC  (see \cite{BHR}).

The goal of this paper is to show:
\begin{thm} \label{thm:main-intro}
There exist explicit good locally testable codes (for every rate $ \rho <1$).
\end{thm}

Remarkably, combining the theorem with earlier works \cite{KMRS,GKORS}, we get as a corollary that the result in the main theorem holds for all $\rho ,\delta$ which satisfy the Gilbert-Varshamov bound \cite{G,V}, namely such that $\rho + h(\delta) < 1$, where $h(x) = -x\log_2 x + -(1-x)\log_2(1-x)$, is the binary entropy function. 
In other words, for every $\rho$ and $\delta$ for which binary error-correcting codes are known to exist, there also exist locally testable codes. 

\begin{cor} \label{cor:main-intro}
For every $0<\rho,\delta<1$ such that $\rho + h(\delta)<1$, there exist good locally testable codes of rate at least $\rho$ and normalized distance at least $\delta$.
\end{cor}

The reader is referred to Theorem \ref{thm:main} and Corollary \ref{cor:main} in Section \ref{sec:proofs} for a definitive formulations of Theorem \ref{thm:main-intro} and Corollary \ref{cor:main-intro}, respectively.

We note that, unlike Theorem \ref{thm:main-intro}, the codes from Corollary \ref{cor:main-intro} are randomized and not completely explicit.
We leave it as an open problem to find a deterministic construction of such binary locally testable codes approaching the Gilbert-Varshamov bound.

Let us now briefly give the construction of our locally testable codes.
Recall that a finite $r$-regular graph $X = (V,E)$ is called a $\lambda$-expander, for $0 < \lambda < 1$, if its normalized second largest eigenvalue $\lambda(X)$ satisfies $\lambda(X) \leq \lambda$.
The expander code of \cite{SS} was constructed as a subspace of the space of functions $\bF_2^E = \{f\,:\,E\rightarrow \bF_2\}$.
The local view of $f$ at $v \in V$, denoted $f_v$, is a vector in $\bF_2^r$, after fixing some bijection from the edges touching $v$ to $\{1,\ldots,r\}$.
The code $C$ is the space of all those functions whose local view $f_v$, from every vertex $v \in V$, is inside a well chosen fixed ``small code" $C_0$ in $\bF_2^r$.
By a simple constraint counting argument, if $\rho(C_0) > \frac{1}{2}$ then $\rho(C) \geq 2\rho(C_0) - 1$, and by using the $\lambda$-expansion of $X$, if $\delta(C_0) > \lambda$ then the local distance of $C_0$ is propagated to give a global distance for $C$.

Our codes are built as a second floor above the expander codes.
More specifically, let $G$ be a finite group with two symmetric sets of generators $A$ and $B$ of size $|A| = |B| = r$.
We look now at the left/right Cayley complex $\mbox{Cay}^2(A;G;B) = (V,E,S)$ which is defined as follows.
Its set of vertices is $V= G$, its set of edges is $E= E_A \cup E_B$, where $E_A = \{ \{g,ag\} \,:\, g\in G,\, a\in A\}$ and  $E_B = \{ \{g,gb\} \,:\, g\in G,\, b\in B\}$, and its set of squares is $S = \{ [a,g,b] \,:\, g\in G,\, a\in A,\, b \in B \}$, where $[a,g,b]$ is the square consisting of edges $\{g,ag\}$, $\{ag,agb\}$, $\{agb,gb\}$, $\{gb,g\}$.
Further details on the left/right Cayley complexes are given in Section \ref{sec:expanders}.

We choose a fixed small code $C_1 \leq \bF_2^r \cong \bF_2^A \cong \bF_2^B$.
This code defines an intermediate code $C_0 = C_1 \otimes C_1 \leq \bF_2^A \otimes \bF_2^B \cong \bF_2^{A \times B}$, the tensor code.
Now look at the space of functions on the squares $\bF_2^S = \{f\,:\,S\rightarrow \bF_2\}$.
Our global code $C$ is defined as the subspace of those functions whose local view $f_e$ from every edge $e$ is in $C_1$.
Note that $C$ can be defined also as the subspace of those functions whose local view $f_v$ from every vertex $v$ is in $C_0 = C_1 \otimes C_1$.
Further details on the left/right Cayley expander codes are given in Section \ref{sec:codes}.

Now, we prove that for a well chosen $C_1$; having sufficiently large rate and normalized distance as in \cite{SS}, but also being smooth in the sense of \cite{DSW}, the tensor code $C_0$ is agreement testable (see Definition \ref{dfn:agreement}).
We then propagate the ``local" local testability of $C_0$ to a ``global" local testability of $C$.
Further details of the above argument are given in Section \ref{sec:proofs}.

This work evolved from the insight provided by Garland's work \cite{Gar}, which shows that HDX (high-dimensional expanders, such as the Ramanujan complexes \cite{LSV1, LSV2}) display  a ``local to global" phenomenon which does not hold for one dimensional graphs.
While our work started from this observation, in the end we used a simpler type of HDX, so this work is quite elementary (with the exception of Theorem \ref{thm:LPS}).
Still the journey through the $p$-adic world left us with some interesting problems. 
This aspect, which is not needed for the proof of the main theorem, will be described in Section \ref{sec:problems}, where some open problems will be suggested.

The current paper is a journal version of the announcement in \cite{DELLM}.
In spite of this it is actually shorter than \cite{DELLM} and contains several changes and improvements.
We eliminated the use of condition (TNC) (see Remark \ref{rem:TNC}) and the use of robust testability (see \cite{DSW}), making the proof more streamlined.
Moreover, rather than using a random construction (for a fixed size) of the base code $C_1$, we replaced it by a fully explicit construction using expander codes (see Proposition \ref{pro:base-code}).
For the sake of simplicity of notations we assume here that $|A| = |B|$, but all arguments still hold without this assumption as in \cite{DELLM}.

Finally, after this work has been completed and announced, we have learnt that Panteleev and Kalachev \cite{PK} proved independently the main result of this paper.
They proved it as a by product to solving the problem of constructing good quantum LDPC codes (see \cite{LZ} for a distinction and comparison between \cite{PK} and our work).
In particular, it is shown in \cite{LZ} how one can reverse the order and use our work to deduce the original result of \cite{PK}, i.e. constructing good quantum LDPC codes.
We also learnt that our notion of left/right Cayley complex can be considered as a special case of the ``balanced product" of $G$-graphs which was defined recently in \cite{BE}.

\subsection{Acknowledgements}
We wish to thank Prahladh Harsha and Avi Wigderson for many interesting discussions along the way of this project. 
We thank Tali Kaufman for her in influential role in connecting LTCs and high dimensional expansion.
We thank Yakov Varshavsky who gave upon our request a semester-long course describing the theory of $p$-adic uniformization (see Section \ref{sec:problems}).

This work started by the first and fourth authors during a year long program at the Israeli Institute of Advanced Studies (IIAS) on high dimensional expanders in 2017.
It was presented by the first author on October 6, 2021, at the Simons Institute for the Theory of Computing [25] as part of the lecture series on breakthroughs in computer science, and at the Institute for Advanced Study in Princeton on October 25-26, 2021 [26]. 
It was also presented by the fourth author on October 27, 2021 at the Simon’s HDX21 workshop [53]. 
The authors are very grateful to these institutions and for the remarks of the audience which improved the exposition of the paper.

Irit Dinur acknowledges support by ERC grant 772839 and ISF grant 2073/21. 
Shai Evra is grateful to the Azrieli Foundation for the award of an Azrieli Fellowship. 
Alexander Lubotzky’s research is supported by the European Research Council (ERC) under the European Union’s Horizon 2020 research and innovation programme (grant agreement No 882751).
Shahar Mozes acknowledges support by ISF-Moked grant 2019/19. 
The first, second and forth authors are indebted to support from the IAS, Princeton.

\section{Expanders} \label{sec:expanders}

Recall the definition of Cayley graphs.
Let $G$ be a finite group and let $A \subset G$ be a symmetric subset (i.e. $a \in A \; \Leftrightarrow \; a^{-1} \in A$) satisfying $1 \not\in A$.
The left Cayley graph of $G$ w.r.t. $A$, denoted by $\mbox{Cay}(A;G)$, is defined to be the finite simple graph whose set of vertices is $V = G$ and whose edges   connect $g$ to $ag$ for all $g\in G$ and $a\in A$. 
We denote an edge connecting $g$ and $ag$ by $\ledge{a,g}$, and this edge is also denoted by $\ledge{a^{-1},ag}$. It is easy to check that the total number of edges is thus $|G|\cdot|A|/2$.
   
A similar definition can be given to the right Cayley graph, denoted $\mbox{Cay}(G;A)$, for which $V = G$ and edges connect $g$ to $ga$ for all $g\in G$ and $a\in A$. 
Here we will similarly denote the edge connecting $g$ to $ga$ by $\redge{g,a}$ or by $\redge{ga,a^{-1}}$.

Next we define the notion of a left/right Cayley complex.

\begin{dfn} \label{dfn:LRCC}
Let $G$ be a group and $A,B \subset G$ be two symmetric subsets.
The left/right Cayley complex of $G$ w.r.t. $A$ and $B$, denoted by $\mbox{Cay}^2(A;G;B)$, is defined to be the following $2$-dimensional square complex:
\begin{itemize}
\item Its set of vertices is $V= G$.
\item The set $E_A$ of left edges and the set $E_B$ of right edges are given by
\begin{equation}
E_A = \{ \ledge{a,g} \,:\, g\in G,\, a\in A\} \qquad \mbox{and} \qquad E_B = \{ \redge{g,b} \,:\, g\in G,\, b\in B\}.
\end{equation} 
The set of edges is the disjoint union $E= E_A \sqcup E_B$. Observe that even if $\ledge{a,g}$ and $\redge{g,b}$ contain the same pair of vertices $\set{g,ag=gb}$, they refer to two distinct edges (one is a ``left'' edge and one is a ``right'' edge).
\item Its set of squares is $S = A \times G \times B / \sim$, where for any $g\in G,\, a\in A,\, b \in B$, 
\begin{equation}
(a,g,b) \sim (a^{-1},ag,b) \sim (a^{-1},agb,b^{-1}) \sim (a,gb,b^{-1}),
\end{equation}
and denote the equivalence class of $(a,g,b)$ by $[a,g,b]$, so
\begin{equation}
[a,g,b] = \left\lbrace (a,g,b) ,\, (a^{-1},ag,b) ,\, (a^{-1},agb,b^{-1}) ,\, (a,gb,b^{-1}) \right\rbrace .
\end{equation}
Given a square $[a,g,b]$, define its set of vertices to be $\{ g, ag, agb, gb \}$, and its set of edges to be $\{ \ledge{a,g},\redge{ag,b},\ledge{a^{-1},agb},\redge{gb,b^{-1}} \}$.
For any square $s \in S$, and any vertex $g\in V$ (resp. edge $e \in E$), we write $g\in s$ (resp. $e \in s$), to mean that $g$ (resp. $e$) is a vertex (resp. an edge) of $s$.
\end{itemize}
\end{dfn}

Note that the graph $(V,E_A)$ is precisely the left Cayley graph $\mbox{Cay}(A;G)$.
Similarly, $(V,E_B)$ is the right Cayley graph $\mbox{Cay}(G;B)$.
The fact that $A$ acts from the left and $B$ acts from the right, gives local commutativity which generates many four-cycles, namely, squares.
\begin{rem} \label{rem:TNC}
Note that all edges of $\mbox{Cay}^2(A;G;B)$ have two distinct vertices since we are assuming $1 \not\in A , B$.
The analogous statement for the squares, i.e. that $[a,g,b]$ has four distinct vertices and four distinct edges, holds if and only if $A$ and $B$ satisfy condition (TNC) 
\begin{equation} \label{eq:TNC} \tag{TNC}
g a \ne b g, \qquad \forall g \in G,\; a\in A,\; b\in B.
\end{equation}
See \cite{DELLM} for explicit constructions of triplets $(G,A,B)$ which satisfy condition (TNC).
\end{rem}

\begin{rem}\label{rem:count}
The number of edges in $Cay^2(A;G;B)$ is exactly $|G|(|A|+|B|)/2$.  Let us count the number of squares.
Each edge $e\in E_A$ participates in $|B|$ squares, and each $e\in E_B$ participates in $|A|$ squares. 
So the total number of pairs $(e,s)$ such that $e\in s$ is $|E_A||B|+|E_B||A|=|G||A||B|$. 
To figure out how many squares there are, we need to know how many edges participate in each square. 
In a typical square $[a,g,b]$ this number is four, but there might be unexpected collisions, as follows:
\begin{itemize}
\item If $g,ag,gb,agb$ are four distinct vertices then the square clearly has four edges.
\item If $ag=gb$, but $g\neq agb$ then the square has three (and not four) vertices, $g,ag=gb,agb$. There are four distinct edges: $\redge{g,b},\ledge{a,g},\redge{agb,b^{-1}},\ledge{a^{-1},agb}$.
\item If $g=agb$, but $ag\neq gb$ then the square has three vertices and four edges, just like in the previous case.
\item If both $ag=gb$ and $agb=g$ then the square has only two vertices: $g=agb$ and $ag=gb$, and only two edges: $\ledge{a,g}=\ledge{a,gb}$ and $\redge{g,b}=\redge{ag,b}$. Observe that in this case necessarily $a^2=1=b^2$.
\end{itemize}
The last case is excluded whenever the following ``no order-$2$ conjugates'' condition, which is weaker than \eqref{eq:TNC}, holds:
\begin{equation}\label{eq:TNC2}\tag{N2C}
\forall a\in A,g\in G,\qquad a^2=1\;\Rightarrow\; g^{-1}ag\not\in B.
\end{equation}
If \eqref{eq:TNC2} holds then the number of squares is  $|G|\cdot|A|\cdot|B|/4$. 
Otherwise the number of squares is at least $|G|\cdot|A|\cdot|B|/4$ but it might be (slightly) larger. 
The reason is that every vertex participates in $|A|\cdot|B|$ distinct squares but every square contains either $2$ or $4$ edges.
\end{rem}

The left/right Cayley complexes are examples of two-dimensional cubical complexes.
Cubical complexes are well-studied, and in particular there are constructions of Ramanujan cubical complexes of bounded degree in any dimension (see \cite{JL}), whose walk dynamics was studied in \cite{Moz}. 
The left/right Cayley complexes have an additional matching labels feature that other complexes are not known to have.

\begin{dfn} \label{dfn:label}
Let $\mbox{Cay}^2(A;G;B) = (V,E,S)$ be a left/right Cayley complex.

For any vertex $g \in V$, define its squares to be $S(g) = \{[a,g,b] \,:\, a\in A,\, b\in B \}$, and define the labelling map $\iota_g \,:\, A \times B \rightarrow S(g)$, $\iota_g(a,b) = [a,g,b]$.

For every edge $\ledge{a,g} \in E_A$ (resp. $\redge{g,b} \in E_B$), define its squares to be $S(\ledge{a,g} ) = \{ [a,g,b] \,:\, b\in B \}$ (resp. $S(\redge{g,b} ) = \{ [a,g,b] \,:\, a\in A \}$), and define the labelling map $\iota_{\ledge{a,g} } \,:\, B \rightarrow S(\ledge{a,g} )$, $ \iota_{\ledge{a,g}}(b) = [a,g,b]$ (resp. $\iota_{\redge{g,b} } \,:\, A \rightarrow S(\redge{g,b} )$, $\iota_{\redge{g,b} }(a) = [a,g,b]$).
\end{dfn} 

\begin{rem}\label{rem:ML}
For every vertex $g\in G$, the $b$-th square of the $a$-th neighbor equals the $a$-th square of the $b$-th neighbor (and both are equal to $[a,g,b]$). 
In symbols,
\begin{equation}
\iota_{\redge{a,g}}(b) = \iota_{\ledge{g,b}}(a).
\end{equation}
Let us call this the {\em matching labels} property. 

The cubical complexes mentioned above  \cite{JL} exist in all dimensions, but they lack a matching-labels property. 
It is an interesting open question whether there are (expanding, constant degree) cubical complexes of dimension above two with a similar type of matching-labels property.
\end{rem}

\begin{lem} \label{lem:label}
The labelling maps in Definition \ref{dfn:label} are well defined and surjective.
The labelling maps are injective if condition \eqref{eq:TNC} holds.
\end{lem}

\begin{proof}
The map $\iota_g$ is clearly well defined and the map $\iota_{\ledge{a,g}} = \iota_{\ledge{a^{-1},ag}}$ (resp. $\iota_{\redge{g,b}}= \iota_{\redge{gb,b^{-1}}}$) is well defined since $[a,g,b] = [a^{-1},ag,b]$ (resp. $[a,g,b] = [a,gb,b^{-1}]$), i.e. the choice of the root vertex $g$ or $ag$ (resp. $g$ or $gb$) does not change the image of the map.
The maps are surjective since for a fixed $g \in V$, the squares $[a,g,b]$ running over all $(a,b) \in A \times B$, give all possible squares which contain $g$. 
Similarly for a fixed $\ledge{a,g} \in E_A$, (resp. $\redge{g,b} \in E_B$), the squares $[a,g,b]$, running over all $b\in B$, (resp. $a \in A$), give all possible squares which contain $\ledge{a,g}$, (resp. $\redge{g,b}$).
Finally the maps are injective if condition \eqref{eq:TNC} holds, since if $\iota_g(a,b) = \iota_g(a',b')$ (resp. $\iota_{\ledge{a,g}}(b) = \iota_{\ledge{a,g}}(b')$, resp. $\iota_{\redge{g,b}}(a) = \iota_{\redge{g,b}}(a')$), then by condition (TNC), $ag=a'g$ and $gb=gb'$ (resp. $gb=gb'$, resp. $ag=a'g$), hence $a=a'$ and $b=b'$ (resp. $b=b'$, resp. $a=a'$), which proves the injectivity.
\end{proof}

The rest of this section will focus on expansion properties of Cayley graphs and complexes.

Let $X = (V,E)$ be an $r$-regular finite simple graph. 
Let $\bR^V = \{f\,:\, V \rightarrow \bR \}$ be the real vector space of functions on the vertices of the graph, and let $\langle , \rangle \,:\, \bR^V \times \bR^V \rightarrow \bR$ be the inner product $\langle f , g \rangle = \sum_{v\in V} f(v)g(v)$.
Let $T_X \,:\, \bR^V \rightarrow \bR^V$ be the normalized adjacency operator of $X$, defined by $T_X f(v) = \frac{1}{r} \sum_{\{u,v\}\in E} f(u)$.
Denote by $\lambda(X) $ the second largest eigenvalue of $T_X$.

\begin{dfn} \label{dfn:exp}
Call $X$ (a one sided) $\lambda$-expander if $\lambda(X) \leq \lambda$, for $0 < \lambda < 1$.
\end{dfn}

Each left/right Cayley complex is comprised of two Cayley graphs, and we call the complex an expander if both of theses graphs are expanders.

\begin{dfn} \label{dfn:LRCC-exp}
$\mbox{Cay}^2(A;G;B)$ is a $\lambda$-expander if both $\mbox{Cay}(A;G)$ and $\mbox{Cay}(G;B)$ are $\lambda$-expanders.
\end{dfn}

Note that if $\mbox{Cay}^2(A;G;B)$ is a $\lambda$-expander, then its underlying graph is also a $\lambda$-expander.
Indeed the normalized adjacency operator of the underlying graph of $\mbox{Cay}^2(A;G;B)$ is $T = \frac{|A|T_A + |B|T_B}{|A|+|B|}$, where $T_A$ and $T_B$ are the normalized adjacency operators of $\mbox{Cay}(A;G)$ and $\mbox{Cay}(G;B)$, respectively, which implies $\lambda(T) \leq \max\{\lambda(T_A) , \lambda(T_B)\}$.

\begin{dfn} \label{dfn:Markov}
Let $\mV$ be a finite set, $\bR^\mV = \{ f\,:\, \mV \rightarrow \bR\}$ the vector space of real functions on $\mV$ and  $\langle , \rangle \,:\, \bR^\mV \times \bR^\mV \rightarrow \bR$, $\langle f,g \rangle = \sum_{v\in \mV} f(v)g(v)$ the standard inner product.
Let $\mM \,:\, \bR^\mV \rightarrow \bR^\mV$ be a linear operator and let $0 < \lambda < 1$.
\begin{itemize}
\item Say that $\mM$ is symmetric if $\langle \mM f , g \rangle = \langle f , \mM g \rangle$, for any $f,g \in \bR^\mV$.
\item Say that $\mM$ is Markov if it is non-negative (namely $\mM f\geq 0$ whenever $f\geq 0$), and if $\mM 1_\mV = 1_\mV$, where $1_\mV \in \bR^\mV$ is the constant one function.
\item Say that $\mM$ is $\lambda$-expanding if $\langle \mM f , f \rangle \leq \lambda \langle f , f \rangle$, for any $f \in \bR^\mV$ such that $\langle f , 1_\mV \rangle = 0$.
\end{itemize}
\end{dfn}

The following Lemma is a generalization of the Alon-Chung Lemma (\cite{AC}) for any symmetric, Markov $\lambda$-expanding operator.

\begin{lem} \label{lem:AC}
Let $\mM \,:\, \bR^\mV \rightarrow \bR^\mV$ be a symmetric, Markov, $\lambda$-expanding linear operator.
Then for any $0 < \delta \leq 1$ and any $R\subset \mV$, letting $1_R$ be the indicator function of $R$,
\begin{equation}
\langle \mM 1_R , 1_R \rangle \geq \delta |R| \qquad \Rightarrow \qquad |R| \geq (\delta - \lambda) |\mV|.
\end{equation}
\end{lem}

\begin{proof}
Write $1_R = \frac{|R|}{|\mV|} 1_\mV + f_R$ and note that $\langle f_R , 1_\mV \rangle =0$.
Then 
\begin{multline}
\delta |R| \leq \langle \mM 1_R , 1_R \rangle = \langle \mM (\frac{|R|}{|\mV|} 1_\mV + f_R) , \frac{|R|}{|\mV|} 1_\mV + f_R \rangle = \frac{|R|^2}{|\mV|} + \langle \mM f_R , f_R \rangle \\
\leq \frac{|R|^2}{|\mV|} + \lambda \langle f_R , f_R \rangle \leq \frac{|R|^2}{|\mV|} + \lambda \langle 1_R , 1_R \rangle =  \frac{|R|^2}{|\mV|} + \lambda |R|.
\end{multline}
Hence $\delta - \lambda \leq \frac{|R|}{|\mV|}$, which completes the proof.
\end{proof}

Next we shall define several adjacency operators on the edges (Definitions \ref{dfn:M-I}, \ref{def:gpar} and \eqref{eq:mpl}) of a $\lambda$-expanding left/right Cayley complex $\mbox{Cay}^2(A;G;B) = (V,E,S)$, prove that several of them are symmetric, Markov, and $\lambda$-expanding (Lemmas \ref{lem:M-I} and \ref{lem:M-II-c}), and finally prove the main claim of this section (Proposition \ref{pro:M}), which will be used in Section \ref{sec:proofs}. 

First let us introduce a bit of notation. 
We let $L= (A\times \set{0}) \sqcup (B\times \set{1})$ so $|L|=2r$ and corresponds bijectively to the edges leaving a vertex $g\in G$ (without us having to worry about potential collisions of names). 
For $\ell\in L$ and $g\in G$, denote 
\begin{equation}
g^\ell = \begin{cases} a g  & \;\ell = (a,0) \\ g b & \;\ell = (b,1) \end{cases}.
\end{equation}
We write $type(\ell)$ to denote the second component in $\ell$ which indicates if we are in $E_A$ or $E_B$, and we denote by $\edge{g;\ell}$ the edge $\ledge{a,g}$ if $\ell=(a,0)$ and the edge $\redge{g,b}$ if $\ell=(b,1)$.
Also, for $\ell=(\ell_1,\ell_2)$, denote $\ell^{-1}=(\ell_1^{-1},\ell_2)$.

\begin{dfn} \label{dfn:M-I}
Define the following normalized adjacency operators,
\begin{equation}
T \,:\, \bR^V \rightarrow \bR^V, \qquad Tf(g) = \frac{1}{2r} \sum_{\ell\in L} f(g^\ell) ,
\end{equation}
the normalized adjacency operator of the underlying graph $(V,E)$,
\begin{equation}
D \,:\, \bR^E \rightarrow \bR^V, \qquad Df(g) = \frac{1}{2r} \sum_{\ell \in L} f(\edge{g;\ell}),
\end{equation}
the normalized unsigned boundary operator from the edges to vertices, and let $D^t$ be the transpose of $D$, which is the normalized unsigned coboundary operator from vertices to edges, i.e.
\begin{equation}
D^t \,:\, \bR^V \rightarrow \bR^E, \qquad D^t f (\edge{g;\ell}) = \frac{1}{2}(f(g) + f(g^\ell)).
\end{equation}
Finally, let 
\begin{equation}\label{eq:M}
M \,:\, \bR^E \rightarrow \bR^E, \qquad M = D^t \circ T \circ D.
\end{equation} 
\end{dfn}

\begin{lem} \label{lem:M-I}
The operator $M$ from Definition \ref{dfn:M-I} is symmetric, Markov and $\lambda$-expanding.
\end{lem}

\begin{proof}
The operator $M = D^t T D$ is symmetric since $(AB)^t = B^t A^t$ and $T$ is symmetric. It is Markov since it is clearly non-negative and since $D 1_E = 1_V$, $T 1_V = 1_V$ and $D^t 1_V = 1_E$. It is $\lambda$-expanding since for any $f \perp 1_E$,
\begin{equation}
\langle M f , f \rangle = \langle D^t T D f , f \rangle = \langle T Df , Df \rangle \leq \lambda \langle Df , Df \rangle \leq  \lambda \langle f , f \rangle,
\end{equation}
where we have used the fact that $\langle Df , 1_V \rangle = \langle f , D^t 1_V \rangle = \langle f , 1_E \rangle =0$, the fact that $T$ is $\lambda$-expanding and the fact that $\|D\| \leq 1$. 
\end{proof}

\begin{dfn}\label{def:gpar}
Define an auxiliary graph $\gp$ whose set of {\em vertices} is $E$, such that each $\edge{g;\ell}\in E$ is connected to  $\edge{g^{\ell'};\ell}$ for all $\ell'\in L$ such that $type(\ell')\neq type(\ell)$. More explicitly, an edge $\ledge{a,g}$ is connected to $\ledge{a,gb}$ for all $b\in B$, and an edge $\redge{g,b}$ is connected to $\redge{ag,b}$ for all $a\in A$.
Let $\mpar\,:\,\bR^E\to \bR^E$ be the normalized adjacency operator of this graph.
\end{dfn}

For $\ell\in L$ define the subset of $\ell$ labelled edges to be
\begin{equation}
E_\ell = \{ \edge{g;\ell} \in E \,:\, g\in G \},
\end{equation}
and observe that $E_\ell = E_{\ell^{-1}}$. 
Moreover,

\begin{lem} \label{lem:M-II-c}
For each $\ell\in L$ the graph $\gp$ has a connected component $\gp(E_\ell)$ with vertex set $E_\ell$. Let $M_\ell^\parallel$ be the normalized adjacency operator of $\gp(E_\ell)$. This operator is symmetric, Markov and $\lambda$-expanding, for any $\ell\in L$. It is explicitly given by
\begin{equation}\label{eq:mpl}
M_\ell^\parallel  f(\edge{g;\ell}) = \frac 1 r \sum_{\stackrel{\ell'\in L}{type(\ell')\neq type(\ell)}} f(\edge{g^{\ell'};{\ell}}).
\end{equation}
Moreover, whenever $\ell\neq \ell^{-1}$ the graph $\gp(E_\ell)$ is isomorphic to the Cayley graph $\mbox{Cay}(A;G)$ or $\mbox{Cay}(G;B)$ depending on $type(\ell)$; and whenever $\ell=\ell^{-1}$ and writing $\ell=(\ell_1,\ell_2)$ the graph $\gp(E_\ell)$ is isomorphic to the Schreier graph $\mbox{Sch}(A;G/\langle \ell_1 \rangle)$ or $\mbox{Sch}(G/\langle \ell_1 \rangle ;B)$ depending on  $type(\ell)$.  
\end{lem}

\begin{proof}
For an edge $e=\edge{g;\ell}$, we call $\set{\ell,\ell^{-1}}$ the label of the edge. 
It is clear from the definition that in $\gp$ an edge is only connected to edges with the same label. So clearly there are no edges connecting elements in $E_{\ell}$ with $E_{\ell'}$ as long as $\ell'\not \in\set{ \ell,\ell^{-1}}$. 

We first prove the ``moreover'' part.
For $\ell\neq \ell^{-1}$ the graph isomorphism is given by the bijection $\edge{g;\ell}\leftrightarrow g$. For $\ell= \ell^{-1}$ the graph isomorphism is given by the map $\edge{g;\ell}\leftrightarrow \set{g,g^\ell}$ which is a bijection between $E_\ell$ and $G/<\ell_1>$. Indeed, one can check that whenever there is an edge $\set{\edge{g;\ell},\edge{g^{\ell'};\ell}}\in E(\gp(E_\ell))$ there is also an edge between $g$ and $g^{\ell'}$ in the corresponding Cayley (or Schreier) graph.

Since the operator $M_\ell^\parallel$ is defined as a normalized adjacency operator of an undirected regular graph, it is clearly symmetric and Markov.
By our assumption the largest second eigenvalue of the Cayley graphs (and therefore any quotient, including the Schreier graph under consideration) is at most $\lambda$. We deduce that $M_\ell^\parallel$ is $\lambda$-expanding.
The explicit formula in \eqref{eq:mpl} is immediate.
\end{proof}

\begin{cor} \label{cor:M-II}
For any $\ell \in L$, the operator $M^\parallel$ from Definition \ref{def:gpar}, preserves the subspace of functions supported on $E_\ell$, which we identify with $\bR^{E_\ell}$.
Moreover, the restriction of $M^\parallel$ to $\bR^{E_\ell}$ coincides with the operator $M_\ell^\parallel$.\qed
\end{cor}

We are now in a position to state the main result of this section.

\begin{pro} \label{pro:M}
Let $\mbox{Cay}^2(A;G;B) = (V,E,S)$ be a $\lambda$-expanding left/right Cayley complex.
Let $0 < \gamma < 1$ and define the operator $M_\gamma = \gamma M + (1 - \gamma) M^\parallel \,:\, \bR^E \rightarrow \bR^E$.
Then for any $0 < \delta < 1$ and any $R \subset E$,
\begin{equation}
\langle M_\gamma 1_R , 1_R \rangle  \geq \delta |R| \qquad \Rightarrow \qquad |R| \geq \frac{\delta - \lambda}{2r}|E|.
\end{equation}
\end{pro}

\begin{proof}[Proof of Proposition \ref{pro:M}]
Since $M_\gamma$ is the convex sum of $M_1 = M$ and $M_0 = M^\parallel$, either $\langle M 1_R , 1_R \rangle  \geq \delta |R|$ or $\langle M^\parallel 1_R , 1_R \rangle  \geq \delta |R|$.
If $\langle M 1_R , 1_R \rangle  \geq \delta |R|$ then by Lemmas \ref{lem:M-I} and \ref{lem:AC}, we get that $|R| \geq (\delta - \lambda)|E|$.
Assume that $\langle M^\parallel 1_R , 1_R \rangle  \geq \delta |R|$.
Let $L' \subset L$ be a set of representatives such that for any $\ell\in L$, exactly one element of $\{\ell,\ell^{-1}\}$ belongs to $L'$.
Denote $R_\ell = R \cap E_\ell$ for any $\ell\in L'$, and note that $R = \bigsqcup_{\ell\in L'} R_\ell$ and $E = \bigsqcup_{\ell\in L'} E_\ell$.
Therefore by Corollary \ref{cor:M-II}, $\langle M^\parallel 1_R , 1_R \rangle = \sum_{\ell\in L} \langle M_\ell^\parallel 1_{R_\ell} , 1_{R_\ell} \rangle$. 
Hence there exists $\ell\in L$, such that $\langle M_\ell^\parallel 1_{R_\ell} , 1_{R_\ell} \rangle \geq \delta |R_\ell|$.
Then by  Lemmas \ref{lem:M-II-c} and \ref{lem:AC}, we get that $|R_\ell| \geq (\delta - \lambda)|E_\ell|$.
Since $|R| \geq |R_\ell|$ and $|E_\ell| \geq \frac{1}{2r} |E|$, we get the claim.
\end{proof}

We end this section with some explicit constructions of very good expander Cayley graphs.
Recall that an $r$-regular graph $X$ is called Ramanujan if $\lambda(X) \leq \frac{2\sqrt{r-1}}{r}$.
In \cite{LPS}, Lubotzky, Phillips and Sarnak gave the first explicit construction of Ramanujan Cayley graphs.

\begin{thm} [\cite{LPS}] \label{thm:LPS}
For any prime $p$, $p \equiv 1 \,(\mbox{mod }4)$, and any prime $q$, $q \equiv 1 \,(\mbox{mod }4p)$, there exist an explicit construction of a symmetric generating subset $S_{p,q} \subset PSL_2(\bF_q)$, of size $p+1$, such that the Cayley graph $\mbox{Cay}(PSL_2(\bF_q),S_{p,q})$ is Ramanujan, i.e.
\begin{equation}
\lambda \left( \mbox{Cay}(PSL_2(\bF_q),S_{p,q}) \right) \leq \frac{2\sqrt{p}}{p+1}.
\end{equation}
\end{thm} 

The following Proposition, which we shall need in Section \ref{sec:proofs}, enables us more freedom in choosing the degrees of the Cayley graphs, at the price of making them only quasi-Ramanujan, i.e. $\lambda$-expander with $\lambda \leq \frac{c \sqrt{r-1}}{r}$, where $c$ is an absolute constant and $r$ is the degree. 

\begin{pro} \label{pro:Cayley-expanders}
Let $p$, $q$ and $S_{p,q} \subset PSL_2(\bF_q)$ be as in Theorem \ref{thm:LPS}.
Then for any $p+1 - \sqrt{p} \leq r \leq p+1$ and any symmetric subset $S \subset S_{p,q}$ of size  $r = |S|$, we have
\begin{equation}
\lambda \left( \mbox{Cay}(PSL_2(\bF_q),S) \right) \leq 5r^{-1/2}.
\end{equation}
\end{pro}

\begin{proof}
Denote by $M = M_{S_{p,q}}$, $M' = M_{S}$ and $M'' = M_{S_{p,q} \setminus S} = M - M'$, the adjacency matrices of $\mbox{Cay}(G,S_{p,q})$, $\mbox{Cay}(G,S)$ and $\mbox{Cay}(G,S_{p,q} \setminus S)$, respectively.
Since $\mbox{Cay}(G,S_{p,q})$ is $|S_{p,q}|$-regular (resp. $\mbox{Cay}(G,S)$ is $|S|$-regular, resp. $\mbox{Cay}(G,S_{p,q} \setminus S)$ is $|S_{p,q} \setminus S|$-regular), the largest eigenvalue of $M$ is $|S_{p,q}|$ (resp. $M'$ is $|S|$, resp. $M''$ is $|S_{p,q} \setminus S|$), with corresponding eigenvector the constant function $1_G$.
Denote by $\lambda = \lambda \left( \mbox{Cay}(G,S_{p,q}) \right)$ and $\lambda' = \lambda \left( \mbox{Cay}(G,S) \right)$ the normalized second largest eigenvalue of $M$ and $M'$, respectively.
By the Courant-Fischer Formula we get  
\begin{equation}
\lambda \cdot |S_{p,q}| = \max_{0 \ne v \perp 1_G} \frac{v^tMv}{v^tv}, \quad \lambda' \cdot |S| = \max_{0 \ne v \perp 1_G} \frac{v^tM'v}{v^tv} \quad \mbox{and} \quad |S_{p,q} \setminus S| = \max_{0 \ne v} \frac{v^tM''v}{v^tv}.
\end{equation}
Therefore 
\begin{equation}
\lambda' \cdot |S| = \max_{0 \ne v \perp 1_G} \frac{v^tM'v}{v^tv} 
\leq \max_{0 \ne v \perp 1_G} \frac{v^tMv}{v^tv} + \max_{0 \ne v} \frac{v^tM''v}{v^tv}
\leq \lambda \cdot |S_{p,q}| + |S_{p,q} \setminus S|,
\end{equation}
and after dividing by $|S|$, and noting that $r \geq p+1 -\sqrt{p}\geq \frac{p}{2}$, we get
\begin{equation}
\lambda' \leq \lambda \frac{|S_{p,q}|}{|S|}   + \frac{|S_{p,q} \setminus S|}{|S|} \leq \frac{2\sqrt{p}}{p+1}\frac{p+1}{r} + \frac{p+1 - r}{r} \leq \frac{3 \sqrt{p}}{r} \leq 5 r^{-1/2}.
\end{equation}
\end{proof}

\section{Codes} \label{sec:codes}

Recall that a (binary, linear) error correcting code $C$ is a subspace of $\bF_2^n$, where $n=n(C)$ is called the block-length of $C$.
Define the rate and the normalized distance of the code to be 
\begin{equation}
\rho(C) = \frac{\dim(C)}{n} \qquad \mbox{and} \qquad \delta(C) = \frac{\min\{\mbox{wt}(v) \,:\, 0 \ne v \in C\}}{n},
\end{equation}
where $\mbox{wt}(v) = |\{1\leq i \leq n \,:\, v_i \ne 0\}|$ is the Hamming weight.
A family of codes is called good if their rates and normalized distances are uniformly bounded away from zero.

Let us now describe the Sipser-Spielman \cite{SS} construction of expander codes for Cayley graphs, building on the work of Tanner \cite{Tan}.
Let $G$ be a group, $A \subset G$ a symmetric subset of size $r = |A|$, $1\not\in A$, and let $\mbox{Cay}(A;G) = (V,E)$ be the (left) Cayley graph.
Let $C_0 \leq \bF_2^A$ be a code of length $r$.
For any $g \in V$, the map $a \leftrightarrow \ledge{a,g}$ gives a bijection between $A$ and $E(g) = \{ e \in E \,:\, g \in e\}$.
For any $f \in \bF_2^E$ and any $g\in V$, define the local view of $f$ at $g$ to be the restriction $f_g = f |_{E(g)} \in \bF_2^{E(g)} \cong \bF_2^A$.
Define the expander code w.r.t. the Cayley graph $\mbox{Cay}(A;G)$ and the small code $C_0$, to be 
\begin{equation} \label{eq:SS-codes}
C[G,A,C_0] = \{f \,:\,  \bF_2^E \;:\; \forall g \in V, \quad f_g \in  C_g \},
\end{equation}
where $C_g\cong C_0$ is defined explicitly by 
\begin{equation}
C_g = \sett{w\in \bF_2^{E(g)}}{w(\ledge{\cdot,g})\in C_0}.
\end{equation}
In their work \cite{SS}, Sipser and Spileman proved the following lower bounds on the rate and normalized distance of the expander codes, in terms of the parameters of the small code and the second largest eigenvalue (i.e. the expansion) of the graph. 

\begin{pro} [\cite{SS}] \label{pro:SS}
Let $\mbox{Cay}(A;G)$ be an $r$-regular $\lambda$-expander Cayley graph and $C_0$ a code of length $r$.
Then the expander code $C[G,A,C_0]$ is of length $n = \frac{r}{2}|G|$, with parameters
\begin{equation}
\rho(C[G,A,C_0]) \geq 2 \rho(C_0) - 1 \qquad \mbox{and} \qquad \delta(C[G,A,C_0]) \geq \delta(C_0) ( \delta(C_0) - \lambda).
\end{equation}
In particular, if $\rho(C_0) > \frac{1}{2}$, $\delta(C_0) > \lambda$ and $\mbox{Cay}(A;G)$ runs over a family of $r$-regular $\lambda$-expander graphs, then the resulting family of expander codes is good.
\end{pro}

For the sake of completeness we give a proof for Proposition \ref{pro:SS}.

\begin{proof}
Let $\rho_0 = \rho(C_0)$, $\delta_0 = \delta(C_0)$ and $C = C[G,A,C_0]$.
The dimension of the code $C = C[G,A,C_0]$ is at least the number of degrees of freedom, $|E|$, minus the number of constraints, $(r - \dim(C_0)) |V| = 2(1-\rho_0)|E|$, hence $\rho(C) = \frac{\dim C[G,A,C_0]}{|E|} \geq 1 - 2(1-\rho_0) = 2\rho_0 - 1$.

Let $0 \ne f \in C$ be such that $\delta(C) = \frac{\mbox{wt}(f)}{|E|}$, where $\mbox{wt}$ denotes the Hamming weight.
Denote $S = \{e \in E \,:\, f(e) \ne 0\}$ and $R = \{g \in V \,:\, \exists e \in S\cap E(g)  \}$.
For any $g \in R$, the local view  $f_g \in \bF_2^{E(g)}$ is a non-zero codeword of $C_0$, hence $\mbox{wt}(f_g) \geq \delta_0 r$, which implies that there are $\delta_0 r$ neighbors of $g$ inside $R$.
Let $T$ be the normalized adjacency operator of $\mbox{Cay}(A;G)$, which is a symmetric, Markov and $\lambda$-expanding.
Then $\langle T 1_R , 1_R \rangle = \sum_{g \in R}  T 1_R(g) =  \sum_{g \in R} \frac{|\{a\in A \,:\, \ledge{a,g} \in E \}|}{r} \geq  \delta_0 |R|$, combined with Lemma \ref{lem:AC}, gives us $|R| \geq (\delta_0 - \lambda)|V|$.
Finally, since each $v\in R$ is contained in at least $\delta_0 r$ edges from $S$, we get $ \delta(C) |E| = \mbox{wt}(f) = |S| \geq \frac{r}{2} \delta_0 |R| \geq \frac{r}{2} \delta_0 (\delta_0 - \lambda) |V| = \delta_0 (\delta_0 - \lambda) |E|$, hence $ \delta(C) \geq  \delta_0 (\delta_0 - \lambda)$.
\end{proof}

Let us now define our construction of left/right Cayley expander codes. Recall that $S(e),S(g)$ are the squares containing an edge $e$ or a vertex $g$ respectively.

\begin{dfn} \label{dfn:codes}
Let $\mbox{Cay}^2(A;G;B) = (V,E,S)$ be a left/right Cayley complex, $r = |A| = |B|$.
Fix a bijection between $A\leftrightarrow B \leftrightarrow \set{1,\ldots,r}$. Let $C_1 \leq \bF_2^r \cong \bF_2^A \cong \bF_2^B$ be a code of length $r$, and define $C_0$ to be its tensor code,
\begin{equation}
C_0 = C_1 \otimes C_1 = \{f \in \bF_2^{A\times B} \;:\; \forall a \in A,\; f(a,\cdot) \in C_1 \quad \mbox{and} \quad \forall b \in B,\; f(\cdot,b) \in C_1 \}.
\end{equation} 
Define for each edge $e\in E$ the local code at $e$ by
\begin{equation}
C_e = \sett{w\in \bF_2^{S(e)} }{ w\circ\iota_e \in C_1}.
\end{equation}
Define for each vertex $g\in G$ the local code at $g$ by
\begin{equation}
C_g = \sett{w\in \bF_2^{S(g)} }{ w\circ\iota_g \in C_0=C_1\otimes C_1}.
\end{equation}
For any $f \in \bF_2^S$ and any edge $e$ define the local view $f_e\in\bF_2^{S(e)}$ of $f$ at $e$ to be the restriction of $f$ to the squares containing $e$.
Define the first left/right Cayley expander code w.r.t. the left/right Cayley complex $\mbox{Cay}^2(A;G;B)$ and the code $C_1$, to be 
\begin{equation}
C[G,A,B,C_1] = \{f \,:\,  \bF_2^S \;:\; \forall e \in E,  \quad f_e \in  C_e\}.
\end{equation}
For any $f \in \bF_2^S$ and any $g \in V$, define the local view $f_g\in\bF_2^{S(g)}$ of $f$ at $g$, denoted $f_g \in \bF_2^{A \times B}$, to be the restriction of $f$ to the squares containing $g$.
Define the second left/right Cayley expander code w.r.t. the left/right Cayley complex $\mbox{Cay}^2(A;G;B)$ and the code $C_0$, to be 
\begin{equation}
C'[G,A,B,C_0] = \{f \,:\,  \bF_2^S \;:\; \forall g \in V,  \quad f_g \in  C_g\}.
\end{equation}
\end{dfn}

First let us note that the two left/right Cayley expander codes are in fact equal. 

\begin{lem} \label{lem:C=C'}
In the notations of Definition \ref{dfn:codes}, $C[G,A,B,C_1]  = C'[G,A,B,C_0]$.
\end{lem}

\begin{proof}
Let $f \in \bF_2^S$ and let $g \in V$, $a\in A$ and $b\in B$. Note that $S(\ledge{a,g}),S(\redge{g,b})\subset S(g)$ and in fact
$\iota_{\ledge{a,g}}(b) = \iota_g(a,b) = \iota_{\redge{g,b}}(a) $ (this comes from the matching labels property, see Remark \ref{rem:ML}).
So for example 
 $f_g\circ\iota_g(a,\cdot)  = f_{\ledge{a,g}}\circ\iota_{\ledge{a,g}}(\cdot)$. 
Hence $f_g \in C_g$ if and only if $f_{e} \in C_e$ for any edge $e \in E$ containing $g$.
This implies that $f_g \in C_g$ for any $g \in V$ if and only if $f_e \in C_e$ for any edge $e \in E$, which proves the claim.
\end{proof}

\begin{rem}
Note that $C[G,A,B,C_1]$ has less constraints in its definition than $C'[G,A,B,C_0]$.
Indeed $C[G,A,B,C_1]$ has $r(1 - \rho(C_1))$ constraints on each edge, and since there are at most $r |G|$ edges, we get at most $(1- \rho(C_1)) r^2 |G|$ constraints,
while $C'[G,A,B,C_0]$ has $r^2(1 - \rho(C_1)^2)$ constraints on each vertex (since $\rho(C_0) = \rho(C_1)^2$), hence a total of  $(1- \rho(C_1)^2) r^2 |G|$ constraints.
This redundancy of (short) constraints should be expected for any LTC as proven in \cite{BGKSV}.
\end{rem}

Next we prove the analogue of Proposition \ref{pro:SS} for left/right Cayley expander codes.

\begin{pro} \label{pro:SS-LRCC}
Let $\mbox{Cay}(A;G;B)$ be a left/right Cayley complex with $r = |A| = |B|$ which is a $\lambda$-expander and let $C_1$ be a code of length $r$.
Then the left/right Cayley expander code $C[G,A,B,C_1]$ is of length $n=|S| \geq \frac{r^2}{4}|G|$, with parameters
\begin{equation}
\rho(C[G,A,B,C_1]) \geq 4\rho(C_1) - 3 \qquad \mbox{and} \qquad \delta(C[G,A,B,C_1]) \geq \frac{1}{4} \delta(C_1)^2 ( \delta(C_1) - \lambda).
\end{equation}
In particular, if $\rho(C_1) > \frac{3}{4}$, $\delta(C_1) > \lambda$ and $\mbox{Cay}^2(A;G;B)$ runs over a family of $r$-regular (i.e. $|A| = |B| = r$) $\lambda$-expanding left/right Cayley complexes, then the resulting family of left/right Cayley expander codes $C[G,A,B,C_1]$ is good.
\end{pro}

\begin{proof}
The dimension of $C[G,A,B,C_1]$ is at least the number of degrees of freedom, $|S|$, minus the number of constraints, $(1 - \rho(C_1)) r |E|$.
Note that $|S| \leq r|E| \leq 4 |S|$, since each edge sits in exactly $r$ squares, i.e. $r|E| = \sum_{e\in E}|S(e)|$, and each square contains at least $1$ and at most $4$ edges (see Remark \ref{rem:count}).
(If condition (TNC) holds then $r|E| = 4|S|$.)
Therefore 
\begin{equation}
\rho(C) = \frac{\dim C[G,A,B,C_1]}{|S|} \geq \frac{|S| - (1 - \rho(C_1)) r |E|}{|S|} \geq  1 - 4(1-\rho(C_1)) = 4\rho(C_1) - 3.
\end{equation}
Let $0 \ne f \in C$ be such that $\delta(C) = \frac{\mbox{wt}(f) }{|S|}$.
Since $f \ne 0$, there exists $e \in E_A$ such that $0 \ne f_e \in C_1$, and if $B' := \{b \in B \,:\, f_e(b) \ne 0 \}$, then $|B'| \geq r \delta(C_1)$.
For each $b \in B'$, define $f^b \in \bF_2^{E_A}$ by $f^b(\ledge{a,g}) = f([a,g,b])$. It is well defined since $f^b(\ledge{a,g}) = f([a,g,b]) = f([a^{-1},ag,b]) = f^b(\ledge{a^{-1},ag})$.
We view $f^b$ as a function on the edges of the Cayley graph $Cay(A;G)$. It is easy to check that on the edges touching a vertex $g$ the restriction of $f^b$ is in $C_1$; so $f^b$ is in the expander code  $C[G,A,C_1]$ defined in \eqref{eq:SS-codes} w.r.t. the Cayley graph $\mbox{Cay}(A;G)$ and the small code $C_1$.
Hence by Proposition \ref{pro:SS}, $\mbox{wt}(f^b) \geq \delta(C_1) (\delta(C_1) - \lambda)|E_A|$, for any $b\in B'$.
Combining all of this, together with the fact that $r|E_A| = \frac{1}{2} r|E| \geq \frac{1}{2} |S|$, we get
\begin{multline}
\delta(C) = \frac{1}{|S|} \mbox{wt}(f) \geq \frac{1}{2|S|} \sum_{b \in B} \mbox{wt}(f^b) \geq \frac{1}{2|S|} \sum_{b \in B'} \mbox{wt}(f^b) \geq \frac{1}{2|S|}  \sum_{b \in B'} \delta(C_1) ( \delta(C_1) - \lambda) |E_A| \\
\geq \delta(C_1)( \delta(C_1) - \lambda) \frac{|B'| |E_A|}{2|S|} \geq \delta(C_1)^2( \delta(C_1) - \lambda) \frac{r|E_A|}{2|S|} \geq \frac{1}{4} \delta(C_1)^2( \delta(C_1) - \lambda).
\end{multline}
\end{proof}

In the rest of this section we prove the existence of a base code with sufficiently good rate and normalized distance, and a form of local testability on its tensor code, called agreement testability, defined below.

Let $C \leq \bF_2^r$ be a code of length $r$ and let $C \otimes C \leq  \bF_2^r \otimes \bF_2^r = M_r(\bF_2)$ be its tensor code, where by linearity, $f \in C \otimes C$ if and only if $f(v,\cdot) \in C$ and $f(\cdot,u) \in C$ for any $1 \leq v,u \leq r$.

\begin{dfn}\label{dfn:agreement}
For any two $f, g \in \bF_2^r \otimes  \bF_2^r$ define the following normalized distances,
\begin{equation}
d(f, g)  = \frac{1}{r^2} |\{ (v, u) \,:\, f(v,u) \ne g(v,u) \}| =  \frac{1}{r^2} \mbox{wt}(f - g),
\end{equation} 
\begin{equation}
d_{row}(f,g) = \frac{1}{r}|\{v \,:\, f(v,\cdot)  \ne g(v,\cdot)\}| \quad \mbox{and} \quad d_{col}(f,g) = \frac{1}{r} |\{ u \,:\, f(\cdot,u) \ne g(\cdot,u) \}|.
\end{equation}
Define the row-column distance of the pair $(f,g)$ from the tensor code $C \otimes C$ to be
\begin{equation}
d_{rc}\left( (f,g) , C\otimes C \right) = \frac{1}{2} \min_{w\in C\otimes C} \left( d_{row}(f,w) + d_{col}(g,w) \right).
\end{equation}
Define the agreement testability parameter $\sigma(C)$ of the tensor code of $C$ to be, 
\begin{equation}\label{eq:def-agr}
\sigma(C) = \min \left\lbrace \frac{d(f, g)}{d_{rc}\left( (f,g) , C\otimes C \right)} \,:\,  f \in \bF_2^r \otimes C ,\, g \in C \otimes \bF_2^r ,\, f \ne g \right\rbrace,
\end{equation}
\end{dfn}

\begin{lem} \label{lem:sigma <=2}
For any code $C$,  $\sigma(C) \leq 2$.
\end{lem}

\begin{proof}
We wish to show that $\sigma(C) \leq 2$, i.e. that $d(f, g) \leq \min_{w\in C\otimes C} \left( d_{row}(f,w) + d_{col}(g,w) \right)$, for any $f, g \in \bF_2^r \otimes  \bF_2^r$.
Note that $d_{row}(f,g) = d_{row}(f-g,0)$,  $d_{col}(f,g) = d_{col}(f-g,0)$ and $d(f,g) = d(f-g,0)$.
In particular, if $\omega \in C \otimes C$ is such that $\left( d_{row}(f,\omega) + d_{col}(g,\omega) \right) = \min_{w\in C\otimes C} \left( d_{row}(f,w) + d_{col}(g,w) \right)$, then by setting $f' = f- \omega$ and $g' = g - \omega$, it suffices to prove $d(f',g') \leq d_{row}(f',0) + d_{col}(g',0)$.
By the triangle inequality, $d(f',g') \leq d(f',0) + d(g',0)$, hence it suffices to prove $d(f,0) \leq  d_{row}(f,0)$ and $d(f,0) \leq  d_{col}(f,0)$, for any  $f \in \bF_2^r \otimes  \bF_2^r$.
We shall prove $d(f,0) \leq  d_{row}(f,0)$, the proof of the other bound follows analogously.
Let $t = |\{1\leq i \leq r \,:\, f(i,\cdot)  \ne 0\}|$, hence $d_{row}(f,0) = \frac{t}{r}$, and note that the non-zero coordinates of $f$ must be contained in the $t$ non-zero rows of $f$, i.e. $ \mbox{wt}(f) \leq tr$, and therefore $d(f, g) \leq \frac{tr}{r^2} = \frac{t}{r} = d_{row}(f,0)$, as claimed.
\end{proof}

The following proposition guarantees the existence of a base code with sufficiently good rate, normalized distance, and agreement testability parameters.

\begin{pro} \label{pro:base-code}
For any $\epsilon > 0$, there exists $\delta_1,  \sigma_1 > 0$, and an infinite family of explicitly constructed codes $\{ C_i \leq \bF_2^{r_i} \}_i$, of even lengths $r_i \rightarrow \infty$, such that for any $i$, 
\begin{equation}
\rho(C_i) \geq 1 - \epsilon,\qquad \delta(C_i) \geq \delta_1, \qquad \mbox{and} \qquad \sigma(C_i) \geq \sigma_1.
\end{equation}
\end{pro}

In \cite[\S~5.1]{DELLM} it was shown that most random LDPC codes satisfy the proposition.
Here we shall give a more explicit proof, by using expander codes.
Both \cite{DELLM} and our proof relies on the notion of smooth codes introduced and studied in \cite{DSW} and \cite{BV}.

Let us fix some notations.
For any $r\in \bN$, denote $[r]=\{1,\ldots,r\}$, $\bF_2^r = \{f \,:\, [r] \rightarrow \bF_2 \}$ and $\cdot \,:\, \bF_2^r \times \bF_2^r \rightarrow \bF_2$, $ f \cdot g = \sum_{k \in [r]}f(k)g(k)$.
For any $H\subset \bF_2^r$ and $d \leq r$, denote $H^\perp = \{g\in \bF_2^r \,:\, \forall f \in H, \; f\cdot g = 0 \}$, $H_{\leq d} = \{f\in H \,:\, \mbox{wt}(f)\leq d \}$ and $H_{\leq d}^\perp = \{g\in H^\perp \,:\, \mbox{wt}(g)\leq d \}$.

If $C$ is a code then $C^\perp$ is its set of constraints and $C_{\leq d}^\perp$ is its subset of short (of weight at most $d$) constraints.
Recall that a code $C \leq \bF_2^r$ is called a $d$-LDPC (low density parity check) code if $C = (C^\perp_{\leq d})^\perp$, i.e. it is defined by its short constraints.

Below we give a slight strengthening of the definition of the notion of smooth codes from \cite{DSW}, which we call uniformly smooth codes.
(In the original definition of \cite{DSW}, the upper bound on the set $J$ is proportional only to $r$, as opposed to $|I|$.)

\begin{dfn} \label{dfn:smooth}
Let $0< \alpha,\beta,\delta <1$, $d\leq r \in \bN$ and $C \leq \bF_2^r$ a code. 
For any $I, J \subset [r]$, denote 
\begin{equation}
C_{\leq d}^\perp(I) = \{ f\in C_{\leq d}^\perp \,:\, f|_I \equiv 0 \} \quad \mbox{and} \quad C(I,J) = \{f|_{[r] \setminus J} \,:\, f\in (C_{\leq d}^\perp(I))^\perp \}.
\end{equation}
The code $C$ is called $(\alpha,\beta,\delta,d)$-US (uniformly smooth), if $C$ is a $d$-LDPC code and
\begin{equation}
\forall I \subset [r], \quad |I| \leq \alpha r, \quad \exists J \subset [r], \quad I\subset J, \quad |J| \leq \beta^{-1} |I| \qquad \mbox{with} \qquad \delta(C(I,J)) \geq \delta.
\end{equation}
\end{dfn}

Let us spell out the definition. 
If $C$ is a $d$-LDPC code then $C_{\leq d}^\perp$ is the set of short constraints of $C$.
Then $C_{\leq d}^\perp(I) \subset C_{\leq d}^\perp$ is the subset of short constraints supported on $[r] \setminus I$, hence $(C_{\leq d}^\perp(I))^\perp$ is a code that contains $C = (C_{\leq d}^\perp)^\perp$, and $C(I,J)$ is obtained from this larger code by restricting its codewords to $[r] \setminus J$.
Call $C$ uniformly smooth, if for any small set $I$, there is a small set $J$ (small relative to $I$), such that $C(I,J)$ has good distance.

Observe that if $C$ is $(\alpha,\beta,\delta,d)$-US then $\delta(C) \geq \delta$.
Indeed, if we take $I = \emptyset$, then in the definition of US code $J = \emptyset$ since $|J|\leq \beta^{-1} |I|$, and note that $C(\emptyset,\emptyset) = C$, hence $\delta(C) \geq \delta$.

The following Lemma shows that uniformly smooth codes have tensor codes which are agreement testable.
The proof essentially follows from the work of the first author with Sudan and Wigderson \cite{DSW}, but for the sake of completeness we give it here.

\begin{lem} \label{lem:smooth-agreement}
Let $0< \alpha, \beta, \delta <1$ be such that $\alpha\beta^{-1} < \min\{ \frac{1}{2}, \delta\}$ and let $2 \leq d \in \bN$.
If $C \leq \bF_2^r$ is a $(\alpha,\beta,\delta,d)$-US code, then $\sigma(C) \geq \frac{\alpha\delta}{d}$.
\end{lem}

\begin{proof}
Let $f \in \bF_2^r \otimes C$ and $g \in C \otimes  \bF_2^r$ such that $f \ne g$.
Let $m = f + g \in M_r(\bF_2)$, and for any $i\in [r]$, denote by $m(i,*), m(*,i) \in \bF_2^r$ the $i$-th row and column of $m$, respectively.
Let $\sigma := \frac{\mbox{wt}(m)}{r^2} = d(f,g)$ and assume $\sigma < \sigma_0 := \frac{\alpha\delta}{d}$ (otherwise there is nothing to prove, since $d_{rc}\left( (f,g) , C\otimes C \right) \leq 1$ for any $f$ and $g$).

Let $I_2  = \{i \in [r] \,:\, \mbox{wt}(m(*,i)) \geq \frac{\delta r}{d} \}$.
Note that $|I_2| \leq \frac{d\sigma r^2}{\delta r}$, and in particular $|I_2| < \alpha r$.
Since $C$ is $(\alpha,\beta,\delta,d)$-US, there exists $J_2 \subset [r]$ such that $I_2 \subset J_2 $ and $|J_2| \leq \beta^{-1} |I_2|$ and such that $\delta(C(I_2,J_2))\geq \delta$.
Define $J_1 = \{j\in [r] \,:\, m(j,*)|_{[r]\setminus J_2} \ne 0 \}$ and $S = ([r] \setminus J_1) \times ([r] \setminus J_2)$.

For any $v\in C_{\leq d}^\perp(I_2)$, then $m\cdot v = 0$.
Indeed, $\mbox{wt}(m\cdot v) <\delta r$ (by the definition of $I_2$), $f \cdot v = \sum_{v_i \ne 0} f(*,i) \in C$ (since $f \in \bF_2^r \otimes C$), $m \cdot v = f \cdot v$ (since $m - f = g \in C \otimes \bF_2^r$), and $\delta(C) \geq \delta$ (since $C$ is $(\alpha,\beta,\delta,d)$-US).
Therefore $m(j,*)|_{[r]\setminus J_2} \in C(I_2,J_2)$ for any $j \in [r]$.
By the definition of $J_1$ and since $\delta(C(I_2,J_2))\geq \delta$, we get $\delta(r-|J_2|)|J_1| \leq \sum_{j \in J_1} \mbox{wt}(m(j,*)) \leq \mbox{wt}(m)$, and since $|J_2|\leq \beta^{-1}|I_2|\leq \beta^{-1}\alpha r \leq \frac{r}{2}$, we get $|J_1| \leq \frac{2r \sigma}{\delta}$, and in particular $|J_1| < \frac{2\alpha}{d} r < \delta r$. 

Next we shall prove that there exists $w \in C \otimes C$, such that $f|_S = w|_S = g|_S$.
By the definition of $J_1$, we get that $m|_S = 0$, hence $f|_S = g|_S$.
Denote $C^{\otimes 2} = C \otimes C$, $C^{\otimes 2}_S = C|_{[r] \setminus J_1} \otimes C|_{[r] \setminus J_2}$ and $\mbox{pr}_S \,:\, C^{\otimes 2} \rightarrow C^{\otimes 2}_S$, $\mbox{pr}_S(f) = f|_S$.
If $w \in C^{\otimes 2}$ is such that $\mbox{pr}_S(w)=0$, then $w(i,*)|_{[r] \setminus J_2} \equiv 0$ and  $w(*,i)|_{[r] \setminus J_1} \equiv 0$ for any $i\in [r]$.
Since $|J_1| < \delta r$, $|J_2|  < \delta r$ and $\delta(C) \geq \delta$, we get that $w(i,*) \equiv 0$ and  $w(*,i) \equiv 0$ for any $i\in [r]$, hence $\mbox{pr}_S$ is injective.
Since $\dim C^{\otimes 2}_S \leq \dim C^{\otimes 2}$, we get that $\mbox{pr}_S$ is an isomorphism.
Note that $f|_S = g|_S \in C^{\otimes 2}_S$, and therefore it is the image of a (unique) $w \in C^{\otimes 2}$ under $\mbox{pr}_S$, i.e. $f|_S = w|_S = g|_S$.

Finally, we give an upper bound on $d_{rc}\left( (f,g) , C\otimes C \right) \leq 1$ in terms of $\sigma = d(f,g)$.
Since $f \in \bF_2^r \otimes C$, $g \in C \otimes  \bF_2^r$, $w \in C\otimes C$ and $f|_S = w|_S = g|_S$, ,and since $|J_1| < \delta r$, $|J_2| < \delta r$ and $\delta(C) \geq \delta$,  then $f$ agrees with $w$ on all rows outside $J_1$ and $g$ agrees with $w$ on all columns outside $J_2$, i.e. $d_{row}(f,w) \leq \frac{|J_1|}{r}$, and $d_{col}(g,w) \leq \frac{|J_2|}{r}$.
Since $|J_1| \leq \frac{2}{\delta}\cdot r \sigma$ and $|J_2| \leq \beta^{-1}|I_2| \leq \frac{d}{\beta \delta} \cdot r \sigma$, and since $\sigma_0 < \frac{\beta \delta}{d} < \frac{\delta}{2}$, we get that $\frac{|J_i|}{r} \leq \sigma_0^{-1} \sigma = \sigma_0^{-1} \cdot d(f,g)$, for $i=1,2$.
Therefore,
\begin{equation}
d_{rc}\left( (f,g) , C\otimes C \right) \leq \frac{1}{2} \left( d_{row}(f,w) + d_{col}(g,w) \right) \leq \frac{1}{2} \left( \frac{|J_1|}{r} + \frac{|J_2|}{r} \right) \leq \sigma_0^{-1} \cdot d(f,g).
\end{equation}
\end{proof}

The next Lemma shows that expander codes are uniformly smooth.

\begin{lem} \label{lem:expcode-smooth}
Let $X=(V,E)$ be a $d$-regular graph, for any $v\in V$, let $C_v \leq \bF_2^{E_v}$ be a local code around $v$, where $E_v = \{e\in E \,:\, v\in e \}$, and let $C = \left\lbrace f\in \bF_2^E \;:\; \forall v\in V, \quad f|_{E_v} \in C_v \right\rbrace$.
If $X$ is a $\lambda$-expander,  $\delta_0 = \min_{v\in V} \delta(C_v)$ and $\lambda < \frac{\delta_0}{4}$, then $C$ is an $(\frac{\delta_0}{8d}, \frac{1}{4d}, \frac{\delta_0^2}{8}, d)$-US code. 
\end{lem}

\begin{proof}
First note that $C$ is an expander code w.r.t. the graph $X$ and the local codes $\{C_v\}_{v\in V}$, and since $X$ is $d$-regular the code is a $d$-LDPC code.
Next, recall that by the Alon-Chung Lemma (Lemma \ref{lem:AC}), if $U \subset V$ has an average degree $\kappa d$, where $0\leq \kappa \leq 1$, then $|U| \geq (\kappa - \lambda)|V|$.

Let $I \subset E$ be a subset of edges of size $|I| \leq \frac{\delta_0}{8d} |E|$.
Let $U_0 \subset V$ be the subset of vertices which touch an edge from $I$.
Define $U_1,U_2,\ldots$, iteratively as follows: 
If $U_{i-1}$ is already defined and there exists $v_i \not\in U_{i-1}$ with more than $\frac{\delta_0 d}{2}$ neighbours in $U_{i-1}$, then define $U_i = U_{i-1} \cup \{v_i\}$.
This process must stop after at most $t = |U_0|$ steps. 
Otherwise the set $U_t$ has $2t$ vertices and at least $\frac{\delta_0 d}{2} t$ edges, i.e. an average degree of at least $\frac{\delta_0}{2} d$, which by the Alon-Chung Lemma implies that $|U_t| \geq (\frac{\delta_0}{2} - \lambda)|V| > \frac{\delta_0}{4} |V|$, contradicting the fact that $|U_t| = 2 |U_0| \leq 4|I| = \frac{\delta_0}{2d} |E| = \frac{\delta_0}{4} |V|$.

Let $U \subset V$ be the final set in the above process and let $J = \{ e \in E \,:\, \exists u \in U,\, u\in e \}$. 
Then: (i) $I \subset J$, (ii)  $|J| \leq d|U| \leq 2d|U_0| \leq 4d|I|$, and (iii) $|E_v \cap J| < \frac{\delta_0 d}{2}$, for any $v \in V \setminus U$.

Let $0 \ne f' \in C(I,J)$ be such that $\delta(C(I,J)) = \frac{\mbox{wt}(f')}{|E \setminus J|}$ and let $f \in (C_{\leq d}^\perp(I))^\perp$ be such that $f' = f|_{E\setminus J}$. Let us write $f'(e) = 0$ for all $e\in J$, and set $R = \{v \in V \,:\, f'|_{E_v} \ne 0\}$.
Since $\mbox{supp}(f') \subset E \setminus J$, by the definition of $U_0$, $U$ and $J$, we get that $v \not \in U$, in particular $v \not\in U_0$, which implies $E_v \cap I = \emptyset$, hence the constraints of $C_v$ belongs to $C_{\leq d}^\perp(I)$, therefore $f'|_{E_v} \in C_v$.
By property (iii) of $J$, $|E_v \cap J| \leq \frac{\delta_0 d}{2}$, therefore $\mbox{wt}(f'|_{E_v\setminus J}) \geq \mbox{wt}(f|_{E_v}) - |E_v \cap J| \geq \delta_0 d - \frac{\delta_0 d}{2} \geq \frac{\delta_0 d}{2}$, i.e. $R$ has an average degree of at least $\frac{\delta_0}{2}d$, and by the Alon-Chung Lemma, $|R| \geq (\frac{\delta_0}{2} - \lambda)|V| >  \frac{\delta_0}{4} |V|$.
Hence, $\delta(C(I,J)) = \frac{1}{2|E \setminus J|}\sum_{v\in R} \mbox{wt}(f'|_{E_v\setminus J}) \geq \frac{\delta_0 d|R|}{4|E|} \geq \frac{\delta_0^2 d|V|}{16|E|} > \frac{\delta_0^2}{8}$, which, combined with properties (i) and (ii) of $J$, proves the Lemma.
\end{proof}

Proposition \ref{pro:base-code} now follows from Lemmas \ref{lem:smooth-agreement} and \ref{lem:expcode-smooth}.

\begin{proof}[Proof of Proposition \ref{pro:base-code}]
Let $1 > \epsilon >0$. 
Let $d = 2^m - 1$, where $m > 20 + 4 \lceil \log_2(\epsilon^{-1}) \rceil$ is large enough such that by \cite[Proposition~1.1]{Alo}, there exists an infinite family of graphs $\{X_i = (V_i,E_i)\}$ which are $d$-regular, $\lambda$-expanders with $\lambda \leq 2.1d^{-1/2}$, and by the construction in \cite[Section~2.1]{Alo}, $X_i$ has $|V_i|= \frac{1}{2}q_i(q_i^2-1)$ vertices, where $q_i$ are primes. 
Note that $8 \mid q_i^2 -1$, for any odd prime $q_i$, hence $4 \mid |V_i|$ and $2 \mid |E_i| = \frac{d}{2}|V_i|$, i.e. $r_i = |E_i|$ are even integers.
Let $C_0$ be a primitive narrow-sense BCH binary code of length $d$ and distance parameter $b = 10 \lceil d^{1/2} \rceil$, which by \cite[Chapter~9, Theorem~1]{MS}, $\delta_0 := \delta(C_0) \geq  \frac{b}{d} \geq 10 d^{-1/2}$ and $\rho(C_0) \geq 1 - \frac{mb}{d} \geq 1 - 20 m d^{-1/2} > 1 - 20 d^{-1/4} > 1 - \frac{\epsilon}{2}$,
Then the family of expander codes $\{ C_i = C[X_i,C_0]\}$, satisfy the requirement of Proposition \ref{pro:base-code}.
Indeed for any $i$, by Proposition \ref{pro:SS}, $\rho(C_i) \geq 1 - \epsilon$ and $\delta(C_i) \geq \delta_0(\delta_0 - \lambda) \geq 79d^{-1} =: \delta_1$, since $\lambda < \frac{\delta_0}{4}$.
By Lemma \ref{lem:expcode-smooth}, $C_i$ is $(\frac{\delta_0}{8d},\frac{1}{4d}, \frac{\delta_0^2}{8}, d)$-US, and by Lemma \ref{lem:smooth-agreement}, we get $\sigma(C_i)\geq \frac{\frac{\delta_0}{8d}\frac{\delta_0^2}{8}}{d}  \geq 15 d^{-7/2} =: \sigma_1$.
\end{proof}

\section{Proofs} \label{sec:proofs}

In this section we state and prove the main result of this paper, namely, constructing locally testable codes (LTCs), with constant rate, constant normalized distance and constant query complexity.

\begin{thm} \label{thm:main}
Let $\epsilon > 0$. 
There exist $\delta_\epsilon,  \kappa_\epsilon > 0$ and $q_\epsilon \in \bN$, and an infinite family of explicitly constructed locally testable codes $\{ C_i \}$, of lengths $n(C_i) \rightarrow \infty$, such that for any $i$,
\begin{equation}
\rho(C_i) \geq 1 - \epsilon, \qquad \delta(C_i) \geq \delta_\epsilon, \qquad q(C_i) = q_\epsilon \qquad \mbox{and} \qquad \kappa(C_i) \geq \kappa_\epsilon.
\end{equation}
\end{thm}
Our codes, just like the Sipser-Spielman codes \cite{SS}, come with a linear-time decoding algorithm (see Algorithm \ref{dfn:algorithm} below). 

Before proving Theorem \ref{thm:main}, let us show how to deduce from it, combined with existing knowledge from \cite{KMRS, GKORS}, the following corollary.

\begin{cor} \label{cor:main}
Let $\epsilon, \rho, \delta > 0$ be such that $1 - \epsilon \leq \rho + h(\delta) < 1$.
 There exists $\kappa_\epsilon > 0$ and $q_\epsilon \in \bN$, and an infinite family of locally testable codes $\{ C_i \}$, of lengths $n(C_i) \rightarrow \infty$, such that for any $i$, 
\begin{equation}
\rho(C_i) \geq \rho, \qquad \delta(C_i) \geq \delta, \qquad q(C_i) \leq q_\epsilon \qquad \mbox{and} \qquad \kappa(C_i) \geq \kappa_\epsilon.
\end{equation}
\end{cor}

\begin{proof}
These codes are obtained from the codes in our main theorem via two local-transformation steps. 
The first step, due to \cite{KMRS}, is to apply an expander-based distance-amplification step due to \cite{AEL} to obtain LTCs with large constant alphabet and parameters approaching the Singleton bound. 
The second step, due to \cite{GKORS}, uses Thommesen's method \cite{Tho} of concatenation with a random invertible linear transformation, to get a binary code with rate and distance approaching the Gilbert-Varshamov bound.    
\end{proof}

The codes in Theorem \ref{thm:main} will be the left/right Cayley codes built on the left/right Cayley complexes constructed in previous sections. Recall from Definition \ref{dfn:label} the notations $S(e) = \{s \in S \,:\, s\ni e\}$ for $e \in E$, and $S(g) = \{s \in S \,:\, s\ni  g \}$ for $g \in V$, and recall also the code $C_g \leq \bF_2^{S(g)}$ which is isomorphic to $C_0$, and the code $C_e\leq \bF_2^{S(e)}$ which is isomorphic to $C_1$.

The local tester is defined as follows,
\begin{dfn} \label{dfn:tester}
Let $C = C[G,A,B,C_1] \leq \bF_2^S$ be the left/right Cayley expander code w.r.t. the left/right Cayley complex $\mbox{Cay}^2(A;G;B) = (V,E,S)$, $|A| = |B| = r$, and the base code $C_1 \leq \bF_2^r$. 

Define the tester $T$ as follows:
Given $f \in \bR^S$, pick a uniformly random vertex $g \in V$, read the values of $f$ at all squares touching $g$, namely read $f|_{S(g)}$, and accept if and only if $f|_{S(g)} \in C_g$.
\end{dfn}
\begin{rem}
A local-test with even fewer queries works as well: Choose a uniformly random edge $e$, read the values of $f$ at the squares touching $e$, and accept iff $f|_{S(e)}\in C_e$. The validity of this test follows because of the robust testability of the tensor code combined with the validity of the test above, cf. \cite{BSS}. 
\end{rem}
Define the testability parameter $\kappa(C)$ of $C$ (w.r.t. the tester $T$) to be 
\begin{equation}
\kappa(C) = \min \left\lbrace \frac{D(f)}{{\dist(f,C)} } \,:\,  f \in \bF_2^S \setminus C \right\rbrace,
\end{equation}
where $\dist(f,C) = \min_{w\in C}\dist(f,w)$ 
is the normalized Hamming distance of $f$ from the code $C$ and
\begin{equation}
D(f) = \frac{1}{|V|} |\{g \in V \,:\, f_{S(g)} \not\in C_g \}| = \mathbb{P}[T \mbox{ Rejects } f],
\end{equation}
i.e. the probability that the tester rejects $f$.

\begin{thm} \label{thm:tester}
Let $C = C[G,A,B,C_1]$, where $\mbox{Cay}^2(A;G;B)$ is a $\lambda$-expander, $r = |A| = |B|$, and $C_1$ be a base code of length $r$, such that $\lambda < \frac{\sigma(C_1) \delta(C_1)}{8 + \sigma(C_1)}$.
Then $C$ is a locally testable code, w.r.t. the tester of Definition \ref{dfn:tester}, with query complexity $q(C) = r^2$ and testability parameter
\begin{equation}
\kappa(C) \geq \frac{1}{4r} \left( \frac{\sigma(C_1) \delta(C_1)}{8 + \sigma(C_1)} - \lambda \right) .
\end{equation}
\end{thm}

In order to prove Theorem \ref{thm:tester} we introduce the following correction algorithm.

\begin{alg} \label{dfn:algorithm}
On input $f \,:\, S \rightarrow \bF_2$ perform the following algorithm.
\begin{itemize}
\item[Start:] For each $g$, let $W_g\in C_g$ be chosen to be the closest element to $f|_{S(g)}$  in $C_g$ (break ties arbitrarily), and let $W=(W_g)_{g\in G}$ be the collection of local views. 
Define the number of disagreeing edges in $W$ by
\begin{equation}
\Delta(W) = |\{ e = \edge{g;\ell} \in E \,:\, W_g|_{S(e)} \ne W_{g^\ell}|_{S(e)} \}|.
\end{equation}
\item[Loop:] If there exists some $g\in G$ and a codeword $w\in C_g$ such that replacing $W_g$ by $w$ reduces $\Delta(W)$ then replace $W_g$ accordingly and repeat. If there is no such $g$, continue to the End step.
\item[End:] If $\Delta(W) > 0$, output ``far", otherwise output
\begin{equation}
F \,:\, S \rightarrow \bF_2, \qquad F(s) = W_g(s), \qquad \forall s \in S, \quad \forall g \in s.
\end{equation}
\end{itemize}
\end{alg}

Note that if Algorithm \ref{dfn:algorithm} ends with $\Delta(W) = 0$, then the definition of $F(s)$ does not depend on the choice of $g \in s$ for any $s \in S$.
Indeed, since $\Delta(W) = 0$, for any $s \in S$ and any two vertices $g,g' \in s$, if
they are connected by an edge, say $g'=g^\ell$, then $W_g(s) = W_{g^\ell}(s)$.
If they are not connected by an edge, then there is a third vertex so that $g,g^\ell,(g^\ell)^{\ell'}=g'$ is a length two path contained in $s$ and again $W_g(s) = W_{g^\ell}(s) = W_{g'}(s)$.
\begin{lem} \label{lem:algorithm-time}
Let $W^0=(W^0_g)_{g\in G} $ be the initial collection of local views in the start step of Algorithm \ref{dfn:algorithm}.
The algorithm terminates after at most $\Delta(W^0)$ iterations and
\begin{equation}
\Delta(W^0) \leq 2 D(f) |E|.
\end{equation}
\end{lem}

\begin{proof}
Since each loop step in the algorithm reduces $\Delta(W) \in \bN$ by at least $1$, the number of iterations is upper bounded by $\Delta(W^0)$.
Note that  if $f|_{S(g)} \in C_g$ then $W^0_g = f|_{S(g)}$, so for any edge $e=\edge{g;\ell} \in E$, if both of its endpoints satisfy $f|_{S(g)}\in C_g$ and $f|_{S(g^\ell)} \in C_{g^\ell}$, then
$W^0_g = f|_{S(g)}$ and $(W^0)_{g^\ell} = f|_{S(g^\ell)}$ and so
 $W^0_g|_{S(e)} = f|_{S(e)}  = (W^0)_{g^\ell}|_{S(e)}$.
Therefore, if $e = \edge{g;\ell} \in E$ contributes to the count of $\Delta(W^0)$, then  either $g$ or $g^\ell$ contribute to the count of $D(f) |V|$, and since each vertex can be counted at most $2r$ times this way, we get that $\Delta(W^0) \leq 2 r D(f) |V| = 2 D(f) |E|$.
\end{proof}

\begin{pro} \label{pro:algorithm-succeed}
If Algorithm \ref{dfn:algorithm} outputs $F \in \bF_2^S$, then $F \in C = C[G,A,B,C_1]$ and
\begin{equation}
\dist(f,C) \leq (4 + 8r)\cdot D(f).
\end{equation}
\end{pro}

\begin{proof}
The fact that $F$ is a codeword of $C$ follows from Lemma \ref{lem:C=C'}, together with the construction of $F$, which guarantees that $F_g \in C_g$ for any $g \in V$.

Let $W^0$ and $W^1$ be the initial and final collections $W$ in the algorithm, respectively. 
Denote $V_0 = \{g \in V \,:\, W^0_g \ne f_g \}$ and $V_1 = \{ g \in V \,:\, W^0_g \ne W^1_g \}$.
Note that $W^1_g = F|_{S(g)}$ for any $g \in V$, and that $f|_{S(g)} = F|_{S(g)}$ for any $g \not\in V_0 \cup V_1$, hence $\frac{\mbox{wt}(f - F) }{|S|} \leq \frac{|A||B||V_0 \cup V_1|}{|S|}$, where we have used the fact that each $g\in G$ participates in $|A||B|$ squares.
By definition $|V_0| = D(f) |V|$, and since each iteration of the loop step of the algorithm affects the value of $W$ in exactly one vertex, $|V_1|$ is bounded by the number of iterations of the algorithm, which combined with Lemma \ref{lem:algorithm-time} gives $|V_1| \leq \Delta(W^0) \leq 2 D(f) |E| = 2r D(f) |V|$. 
All in all, using $|S|\geq \frac{|G||A||B|}4$ (see Remark \ref{rem:count}), we get
\begin{equation}
\mbox{dist}(f,C) \leq \frac{\mbox{wt}(f - F)}{|S|} \leq \frac{|A||B||V_0 \cup V_1|}{|S|} \leq \frac{|V_0| + |V_1|}{|V|/4} \leq  (4 + 8r) D(f).
\end{equation}
\end{proof}

\begin{pro} \label{prop:algorithm-fails}
Let $\mbox{Cay}^2(A;G;B)$ be a $\lambda$-expander and denote $\delta_1 = \delta(C_1)$ and $\sigma_1 = \sigma(C_1)$ (as defined in \eqref{eq:def-agr}).
If Algorithm \ref{dfn:algorithm} outputs ``far" on input $f \in \bF_2^S$, then
\begin{equation}
D(f) \geq  \frac{1}{4r} \left( \frac{\sigma_1 \delta_1}{16 + \sigma_1} - \lambda \right).
\end{equation}
\end{pro}

Our proof of Proposition \ref{prop:algorithm-fails} focuses on the set $ R \subset E$ of disputed edges, i.e. the edges which contribute to $\Delta(W)$, for the final $W$ when Algorithm \ref{dfn:algorithm} outputs ``far". We will show that $R$ is large by describing a highly-mixing random walk on the edges and showing that $R$ has a large ``staying probability'' with respect to this random walk. Namely, if we are at $R$ and take a random step we will remain in $R$ with good probability. Standard expansion propagation arguments imply that $R$ takes up a constant fraction of the entire set of edges (see Lemma \ref{lem:AC}). To show that $R$ has a large ``staying probability'' we analyze its local structure, relying on the distance of $C_1$ and on the agreement testability of $C_0=C_1\otimes C_1$. 

Let us begin with some notation.
For a vertex $g\in G$, let 
\begin{equation}
\na(g) = \Card{\sett{\ell\in L}{\edge{g;\ell}\in R}}
\end{equation}
be the number of edges touching $g$ that are in $R$, and similarly for an edge $e=\edge{g;\ell}$ let 
\begin{equation}
\na(\edge{g;\ell}) = \na(g)+\na(g^\ell).
\end{equation}
For an edge $e=\edge{g;\ell}\in E$ we need to count how many edges parallel to $e$ are in $R$. 
So let 
\begin{equation}
\npar(\edge{g;\ell}) = \Card{\sett{\ell'\in L}{type(\ell')\neq type(\ell)\hbox{ and }[g^{\ell'},\ell]\in R}}.
\end{equation}
For example if $e=\redge{g,b}\in E_B$, $\npar(e)$ is the number of $a\in A$ for which $\redge{ag,b}\in R$.
The importance of the quantity $\npar(e)$ stems from the fact that 
\begin{equation}\label{eq:Mpar}
 M^\parallel 1_R(e) = \frac {\npar(e)}r,
\end{equation}
which follows directly from the definition of $M^\parallel$, see Definition \ref{def:gpar}.
\begin{lem} \label{lem:locdist}
Let $\delta_1 = \delta(C_1)$. 
Then for any $e \in R$,
\begin{equation}
\npar(e)+\na(e) \geq  \delta_1 r.
\end{equation}
\end{lem}

\begin{proof}
Since $e = \edge{g;\ell} \in R$, $0 \ne (W_g|_{S(e)} - W_{g^\ell}|_{S(e)}) \in C_e$, we get that $W_g $ and $ W_{g'}$ disagree on at least $\delta_1 r$  squares in $S(e)$.
Suppose $s \in S(e)$ is such that $W_g(s) \ne W_{g^\ell}(s)$, and suppose that $s = [a,g,b]$ for some $a\in A$ and $B\in B$. 
Assume that $e=\ledge{a,g}$ (resp. $e=\redge{g,b}$). 
If all other edges of the square were not in $R$ we would get
$W_g(s) = W_{gb}(s)=W_{agb}(s)=W_{ag}(s)$ contradicting the fact that $W_g(s)\neq W_{ag}(s)$ (resp. we would get that  $W_g(s) = W_{ag}(s)=W_{agb}(s)=W_{gb}(s)$ contradicting the fact that $W_g(s)\neq W_{gb}(s)$).

Therefore, each such square adds at least one to the count of $\npar(e)+\na(e)$ and we get the required inequality.
\end{proof}

The above is a first step in showing that a single edge in $R$ implies more edges in $R$. 
If $\npar(e)$ is large then since the parallel walk is rapidly mixing we are in good shape. 
However, it could be that $\npar(e)$ is small or even zero and only $\na(e)$ is large. 
Unfortunately, the operator that moves from an edge to another edge sharing a vertex, does not mix quickly enough (it gets stuck inside the set of edges touching some vertex $g_0$ with probability $1/2$). 
Instead, we consider the slightly more complicated operator $M$ (recall Definition \ref{dfn:M-I}). 
For this we must introduce some more notation.
For a vertex $g\in G$, let 
\begin{equation}
\m(g) = \Card{\sett{(\ell_1,\ell_2)\in L^2}{\edge{g^{\ell_1};{\ell_2}}\in R}},
\end{equation}
and for an edge $e=\edge{g;\ell}$ let 
\begin{equation}
\m(e) = \m(g)+\m(g^\ell).
\end{equation}
Note that the definition of $\m(e)$ might count some edges twice.
The importance of $\m(e)$ stems from the fact that
\begin{equation}\label{eq:Mop}
 M1_R(e) = \frac {\m(e)}{8r^2}.
\end{equation}
which follows directly from the definition, see Definition \ref{def:gpar}. 
Moreover,
\begin{lem}\label{lem:n1n2} 
For every edge $e\in E$,
\begin{equation}
\frac{\na(e)}{4r}\leq 4\sigma_1^{-1} \cdot \frac{\m(e)}{8r^2}.
\end{equation}
\end{lem}

The proof of this lemma relies on the agreement testability of $C_0$, which implies that if a vertex $g$ touches many $R$ edges, then there must also be many $R$ edges in its ``link''. 
By ``link" we mean the set of edges in squares containing $g$ that are not adjacent to $g$. 
These are exactly edges of the form $\edge{g^{\ell_1};\ell_2}$ where $type(\ell_1)\neq type(\ell_2)$. 
Let $\m'(g)$ count the number of such edges that land in $R$,
\begin{equation}\label{eq:link}
\m'(g) = \Card{\sett{(\ell_1,\ell_2)\in L^2}{type(\ell_1)\neq type(\ell_2)\hbox{ and }\edge{g^{\ell_1};\ell_2}\in R}}.
\end{equation}
Clearly
\begin{equation}\label{eq:mm1} 
\m(g) \geq \m'(g)
\end{equation}
and conveniently, the agreement testability of $C_0$ allows us to relate $\m'(g)$ to $\na(g)$:

\begin{lem}\label{lem:s2}
Let $\sigma_1 = \sigma(C_1)$ be as in Definition \ref{dfn:agreement}.
Then for any $g\in G$,
\begin{equation}\label{eq:ltc}
\frac{\na(g)}{2r} \leq 2\sigma_1^{-1}\cdot \frac{\m'(g)}{2r^2}.
\end{equation} 
\end{lem}

We will prove the above two lemmas shortly, but first let us complete the proof of Proposition \ref{prop:algorithm-fails}.

\begin{proof}[Proof of Proposition \ref{prop:algorithm-fails}]
First recall \eqref{eq:Mpar} and \eqref{eq:Mop} which allow us to express $M^\parallel 1_R(e)$ and $M1_R(e)$ in terms of $\npar(e)$ and $\m(e)$ respectively. 
Together with \eqref{eq:mm1}, Lemma \ref{lem:n1n2}, and Lemma \ref{lem:locdist} we get for each $e\in E$,
\begin{equation} \label{eq:alg-fail-4}
\left( M^\parallel + 16 \sigma_1^{-1} M \right) 1_R (e) = 
\frac{\npar(e)}{r} + 16\sigma_1^{-1}\cdot \frac{\m(e)}{8r^2}\geq 
\frac{\npar(e) + \na (e)}{r} \geq \delta_1.
\end{equation} 
Summing Equation \eqref{eq:alg-fail-4} over all $e \in R$ and multiplying by $\frac{\sigma_1}{16 + \sigma_1}$, gives
\begin{equation} \label{eq:alg-fail-5}
\left\langle \left( \frac{\sigma_1}{16 + \sigma_1} M^\parallel + \frac{16}{16 + \sigma_1}M \right) 1_R , 1_R \right\rangle  
=  \frac{\sigma_1}{16 + \sigma_1} \sum_{e \in R} \left( M^\parallel + 16 \sigma_1^{-1} M \right) 1_R (e) \geq \frac{\sigma_1}{16 + \sigma_1} \delta_1 |R|.
\end{equation} 
Therefore, by applying Proposition \ref{pro:M} together with Equation \eqref{eq:alg-fail-5},
\begin{equation} \label{eq:alg-fail-6}
|R| \geq \frac{1}{2r} \left( \frac{\sigma_1 \delta_1 }{16 + \sigma_1} - \lambda \right) |E|.
\end{equation}
Finally by Lemma \ref{lem:algorithm-time}, the fact that $\Delta(W^0) \geq \Delta(W) = |R|$ and Equation \eqref{eq:alg-fail-6}, 
\begin{equation}
D(f) \geq \frac{|R|}{2|E|} \geq \frac{1}{4r}  \left( \frac{\sigma_1 \delta_1 }{16 + \sigma_1} - \lambda \right).
\end{equation}
\end{proof}

It remains to prove the lemmas.

\begin{proof}[Proof of Lemmas \ref{lem:n1n2} and \ref{lem:s2}]
Fix $g\in G$ and define $w_0,w_1,w_2 \in \bF_2^{A \times B}$, by $w_0(a,b) = W_g([a,g,b])$, $w_1(a,b) = W_{ag}([a^{-1},ag,b])$ and $w_2(a,b) = W_{gb}([a,gb,b^{-1}])$.
Note that the three values are the ``opinions" of the three vertices $g$, $ag$ and $gb$ on the square $[a,g,b]$.
Also note that $w_0 \in C_1 \otimes C_1$, $w_1 \in \bF_2^A \otimes C_1$ and $w_2 \in C_1 \otimes \bF_2^B$. 
Furthermore $w_1(a,\cdot) \ne w_0(a,\cdot)$ if and only if $\ledge{a,g} \in R$ and $w_2(\cdot,b) \ne w_0(\cdot,b)$ if and only if $\redge{g,b} \in R$, so
\begin{equation}
\na(g) = r\cdot(d_{row}(w_1, w_0)+d_{col}(w_2 ,w_0)) .
\end{equation}

Recall the definitions of $d_{row},d_{col},d_{rc}$ from Definition \ref{dfn:agreement}.

Observe that 
$d_{row}(w_1, w_0)+d_{col}(w_2 ,w_0) \leq d_{row}(w_1, w)+d_{col}(w_2 ,w)$ for any $w \in C_0$, since otherwise  Algorithm \ref{dfn:algorithm}  would replace $W_g$ by $w_0$ and decrease $\Delta(W)$, contradicting the fact that the algorithm terminates and outputs ``far". This means that 
\begin{equation}\label{eq:n}
d_{rc}((w_1,w_2),C_1\otimes C_1) = \frac 1 2 (d_{row}(w_1, w_0)+d_{col}(w_2 ,w_0)) = \frac 1 {2r} \na(g).
\end{equation}

Next note that $w_1(a,b) \ne w_2(a,b)$ implies that either $\redge{ag,b} \in R$ or $[gb,a] \in R$, so
\begin{equation}\label{eq:ntwo}
 d(w_1,w_2)\cdot{r^2}\leq \m'(g) 
\end{equation}
Combining \eqref{eq:n} and \eqref{eq:ntwo}  with Definition \ref{dfn:agreement}, gives for any $e \in R$ 
\begin{equation} 
\frac{\na(g)}{2r} =  d_{rc}(w_1 , w_2) \leq \sigma_1^{-1} d(w_1,w_2) \leq  2\sigma_1^{-1}\cdot \frac{ \m'(g)}{2r^2}.
\end{equation}
This completes the proof of Lemma \ref{lem:s2}. To prove Lemma \ref{lem:n1n2} it remains to recall from \eqref{eq:mm1} that $\m'(g) \leq \m(g)$ and together with Lemma \ref{lem:s2} we get for any edge $e=\edge{g;\ell}$, 
\begin{equation}
\frac{ \na(e)}{4r} = \frac{ \na(g)+\na(g^\ell)}{4r} \leq 4\sigma_1^{-1}\cdot \frac{\m'(g)+\m'(g^\ell)}{8r^2} 
\leq 4\sigma_1^{-1}\cdot\frac{\m(g)+\m(g^\ell)}{8r^2}  = \sigma_1^{-1}\cdot\frac{\m(e)}{2r^2}.
\end{equation}
\end{proof}

From Propositions \ref{pro:algorithm-succeed} and \ref{prop:algorithm-fails} we get Theorem \ref{thm:tester}.

\begin{proof}[Proof of Theorem \ref{thm:tester}]
If $f \in C$, then $f_g \in C_0$ for any $g \in V$, hence $\bP[T \mbox{ Accept}] = 1$.

The query complexity of the tester $T$ is $q(C) = r^2$, since for any input $f \in \bF_2^S$ and any random vertex $g \in V$, the tester queries the local view $f_g \in \bF_2^{A \times B}$, at $r^2 = |A \times B|$ values.

Denote $\kappa = \frac{1}{4r} \left( \frac{\sigma(C_1) \delta(C_1)}{16 + \sigma(C_1)} - \lambda \right)$.
Since $\delta(C_1) \leq 1$, $0 \leq \sigma(C_1) \leq 2$ (by Lemma \ref{lem:sigma <=2}), $\lambda \geq 0$ and $r\geq 1$, we get $\kappa \leq (32 r)^{-1} \leq  (4 + 8r)^{-1}$. 
Given $f \in \bF_2^S$, apply Algorithm \ref{dfn:algorithm} to it.
If the algorithm output $F \in \bF_2^S$, we get by Proposition \ref{pro:algorithm-succeed} that
\begin{equation}
\mathbb{P}[T \mbox{ Rejects } f]  \geq (4 + 8r)^{-1} \mbox{dist}(f,C) \geq \kappa \cdot \mbox{dist}(f,C).
\end{equation}
If the algorithm outputs ``far", since $\mbox{dist}(f,C) \leq 1$, we get by Proposition \ref{prop:algorithm-fails} that
\begin{equation}
\mathbb{P}[T \mbox{ Rejects } f]  \geq \kappa  \geq \kappa\cdot  \mbox{dist}(f,C).
\end{equation}
\end{proof}

We are now in a position to prove Theorem \ref{thm:main}.

\begin{proof}[Proof of Theorem \ref{thm:main}]
Let $0 < \epsilon < \frac{1}{2}$.
By Proposition \ref{pro:base-code} there exists $\delta_1, \sigma_1 > 0$, and a code $C_1$ of length $r$, where $r$ is the first even integer $r_i$ which is larger then $10^4 \delta^{-2}\sigma^{-2}$ and such that there exists a prime $r \leq p \leq r+\sqrt{r}$, satisfying $\rho(C_1) \geq 1 - \frac{\epsilon}{4}$, $\delta(C_1) \geq \delta_1$ and $\sigma(C_1) \geq \sigma_1$.
Let $\{q_i\}_{i=2}^\infty$, be the sequence of primes in the arithmetic progression $q_i \equiv 1 \mod{4p}$ (by Dirichlet Theorem this sequence is infinite).
By Theorem \ref{thm:LPS}, for any $i$, there exists a Ramanujan generating subset $S_i \subset G_i = PSL_2(\bF_{q_i})$ of size $p+1$, and by Proposition \ref{pro:Cayley-expanders}, for any symmetric subset $A_i \subset S_i$ of size $r= |A_i|$, the left/right Cayley complex $\mbox{Cay}^2(A_i;G_i;A_i)$ is a $\lambda$-expander for $\lambda = 5 r^{-1/2} \leq \frac{\sigma_1 \delta_1 }{20}$.
By Lemma \ref{lem:sigma <=2}, $\sigma_1 \leq 2$, hence $\lambda \leq \frac{\delta_1}{10}$ and $\frac{\sigma_1 \delta_1}{16 + \sigma_1} \geq \frac{\sigma_1 \delta_1}{18}$.
Define our family of codes to be the left/right Cayley expander codes $\{C_i = C[G_i,A_i,A_i,C_1] \}_{i=2}^\infty$, of lengths $n(C_i) \geq \frac{|A_i|^2|G_i|}{4} = \frac{1}{4} r^2 (q_i^3 - q_i) \rightarrow \infty$.
By Theorem \ref{pro:SS-LRCC}, we get 
\begin{equation}
\rho(C_i) \geq 4 \rho(C_1) - 3 \geq 1 - \epsilon \qquad \mbox{and} \qquad \delta(C_i) \geq \frac{1}{4} \delta_1^2 (\delta_1 -\lambda) \geq \frac{9}{40} \delta_1^3 =: \delta_\epsilon
\end{equation}
and by Theorem \ref{thm:tester} we get
\begin{equation}
q(C_i) = r^2 =: q_\epsilon \qquad \mbox{and} \qquad  \kappa(C_i) = \frac{1}{4r} \left(\frac{\sigma_1 \delta_1}{16 + \sigma_1} - \lambda \right) \geq \frac{\delta_1\sigma_1}{720 r} := \kappa_\epsilon
\end{equation}
which completes the proof of the Theorem.
\end{proof}

\begin{rem} \label{rem:constants}
The constants in Theorem \ref{thm:main} depends on $\epsilon$ poly-logarithmically.
\end{rem}

\section{High dimensional expanders: suggestions for further research} \label{sec:problems}

The current paper is mainly elementary and almost self-contained. 
But it came up as a result of a much longer and intensive journey. 
Some interesting open problems were left aside along the way.  
It is, therefore, worthwhile to give the story here.

Although expander codes are typically not locally testable \cite{BHR} the hope was that higher dimensional versions would be.
This optimistic belief was inspired by local to global expansion behavior of certain high dimensional simplicial complexes that was uncovered already by Garland in his seminal work \cite{Gar}.

In \cite{Gar}, Garland proved a conjecture of Serre, predicting the vanishing of the cohomology of co-compact lattices in high-rank simple $p$-adic groups.  
Equivalently,  if $X$ is a finite simplicial quotient of a Bruhat-Tits building of dimension $d \geq 2$, its cohomology vanishes in dimensions $1 \leq k \leq d-1$. 
The proof of Garland is ``local-to-global'': he showed that if the links of $(d-2)$-dimensional faces have (large) spectral gap, then so do the global Laplacians of $X$. 
Namely, if $X$ is locally an expander, then it is also globally so. 
(For a purely combinatorial treatment and generalizations, see \cite{Opp}). 
Thus, the local spectral gap implies the vanishing of the cohomology. 

This ``local to global'' approach is a high-dimensional phenomenon that does not hold for graphs! 
In graphs, the local structure does not reveal any information about the global expansion. 
To illustrate this, the reader may recall the LPS Ramanujan graphs \cite{LPS} which are $(p+1)$-regular expander graphs with large girth. 
One can easily get $(p+1)$-regular graphs with large girth (and hence locally isomorphic to the LPS ones) which are far from being expanders. 
On the other hand, the Garland method shows that (strong) local expansion implies global expansion in the high dimensional case.

Inspired by all this, the idea was to construct LTC by using the local-to-global behavior of the Ramanujan complexes (\cite{LSV1, LSV2}) in an analog to the way \cite{SS} used Ramanujan graphs for LDPC codes.
For simplicity, we will describe it from now on only in dimension $2$, but one can do the same in higher dimensions.
We shall present two possible constructions of codes on Ramanujan complexes, which we call high dimensional arithmetic codes and high dimensional expander codes.

\subsection{High dimensional arithmetic codes}

Our first code will be defined using the theory of $p$-adic uniformization of Shimura varieties (see \cite{Var}).
Recall that in \cite{Mum}, Mumford used the theory of $p$-adic uniformization and deep result of Yau \cite{Yau}, to construct a (connected component of a) Shimura surface appearing as a locally symmetric quotient of $PU(2,1)$, whose topology is related to the combinatorial structure of a specific quotient of a Bruhat-Tits building of $PGL_3(\bQ_2)$. 
Our plan was to go in the opposite direction and to use the theory of Shimura surfaces to study finite quotients of Bruhat-Tits buildings. 

Let $X = \Gamma \backslash \mB$ be a finite quotient of the Bruhat-Tits building $\mB$ of $PGL_3(\bQ_p)$ by a cocompact torsion-free lattice $\Gamma \leq PGL_3(\bQ_p) = \mbox{Aut}(\mB)$. 
We suppose that $\Gamma$ is given as a congruence subgroup of a projective similitude unitary group as follows.
Let $G = PGU(A,\sigma)$ be a projective similitude unitary algebraic group over $\bQ$ w.r.t. the matrix algebra $A = M_3(E)$, $E/\bQ$ an imaginary quadratic field, assumed for simplicity to be of class number one, such that $p = \p \bar{\p}$ splits in $E$, and an involution of the second type $\sigma(g) = H^{-1} g^* H$, where $H$ is a definite non-degenerate Hermitian matrix.
Then $\Gamma = G(\bQ) \cap K^{p,\infty}$, for some compact open $\{p,\infty\}$-adelic subgroup $K^{p,\infty} \leq G(\bA^{p,\infty})$, where $\bA^{p,\infty} = \prod'_{v \ne p , \infty} \bQ_v$ is the ring of $\{p,\infty\}$-adeles of $\bQ$.
Note that $\Gamma \leq G(\bR \times \bQ_p)$ is a cocompact lattice, and since $G(\bQ_p) \cong PGL_3(\bQ_p)$ and $G(\bR) \cong PU(3)$ is compact, $\Gamma$ is a cocompact lattice in $PGL_3(\bQ_p)$.

By \cite[Section~1]{Mum}, there is a smooth projective scheme $\mX$ defined over $\bZ_p$, satisfying the following.
Let $\bar{\mX} = \mX \,(\,\mbox{mod}\, p\,)$, be the special fiber of $\mX$ defined over $\bF_p$, which is a reducible algebraic surface whose irreducible components are isomorphic to $\tilde{\bP}^2$, the blow up of the projective plane $\bP^2$ at the $p^2+p+1$ rational points (so that each rational point is replaced by a copy of $\bP^1$).
Let $\Delta(\bar{\mX})$ be the dual complex of $\bar{\mX}$, i.e. the simplicial complex whose vertices are the irreducible components of $\bar{\mX}$, two (resp. three) of which are connected by an edge (resp. form a triangle) if their intersection is non-empty.
Then $\Delta(\bar{\mX})$ is isomorphic to the quotient of the Bruhat-Tits building $X = \Gamma \backslash \mB$.

Using the theory of $p$-adic uniformization of Shimura varieties (see \cite{Var, RZ}), $\mX$ can be described in the following way.
Let $G' = PGU(A',\sigma')$ be a projective similitude unitary algebraic group over $\bQ$ w.r.t. $A' = D$ the unique division algebra of degree $3$ over $E$ which ramifies precisely at the two places above $p$, and $\sigma'$ an involution of the second type whose signature at infinity is $(2,1)$.
Note that $G'(\bR) = PU(2,1)$, $G'(\bQ_p) = D_{\p}^*$ is compact and $G'(\bA^{p,\infty}) \cong G(\bA^{p,\infty})$.
Let $\Gamma' = G'(\bQ) \cap K^\infty$, where $K^\infty = K_p K^{p,\infty} \leq G'(\bA^{\infty})$, $K^{p,\infty} \leq G'(\bA^{p,\infty})$ is the level defining $\Gamma$ above and $K_p = \mO_{D_{\p}}^* \leq G'(\bQ_p)$. 
Let $X' = \Gamma' \backslash \mB'$ be the compact quotient of the $2$-dimensional complex unit ball $\mB'$, by $\Gamma' \leq PU(2,1) = \mbox{Aut}(\mB')$.
By the theory of Shimura varieties, there exists an algebraic surface scheme $\mX'$, defined over the ring of integers of some number field, such that the complex manifold $X' = \Gamma' \backslash \mB' $ is a connected component of $\mX'(\bC)$, and such that by taking a base change to $\bZ_p$ we get our previous scheme $\mX$.

In summary, our finite quotient of the Bruhat-Tits building $X = \Gamma \backslash \mB$, can be identified with the dual complex of a connected (since $E$ is of class number one) reducible algebraic Shimura surface $\bar{\mX}/\bF_p$, i.e. $X \cong \Delta(\bar{\mX})$.
For any vertex $v\in X$ (resp. edge $e = \{v_1,v_2\} \in X$), denote by $\bP_v$ (resp. $\bP_e$) the corresponding component in $\bar{\mX}$ (resp. the intersection of $\bP_{v_1}$ and $\bP_{v_2}$).
Note that $\bP_v \cong \tilde{\bP}^2$, the blow up of the projective plane at the $p^2+p+1$ $\bF_p$-rational points, and $\bP_e \cong \bP^1$, the projective line.
Moreover, for any edge $e = \{v_1,v_2\} \in X$, the line $\bP_e$, appears in one of the two components, $\bP_{v_i}$, $i=1,2$, as an 'old' line (from $\bP^2$) and in the other as a 'new' line (i.e. a blow up of a point from $\bP^2$).

Next we define the local codes, i.e. the codes that live on the links of of the complex.

\begin{dfn} \label{dfn:HDAC-local}
(i) Let $\Omega^1 = \Omega^1(\bP^1)$ be the $\bF_p$-vector space of $1$-forms on $\bP^1$ which have only simple poles and only on the points in $\bP^1(\bF_p)$, the set of $p+1$ $\bF_p$-rational points in $\bP^1$.
For any  $t \in \bP^1(\bF_p)$, define the linear map $\mbox{res}_t \,:\, \Omega^1 \rightarrow \bF_p$, $\mbox{res}_t(\omega)$ the residue of $\omega$ at $t$, for any $\omega \in \Omega^1$, and define $\mbox{res}_1 = \oplus_{t \in \bP^1(\bF_p)} \mbox{res}_t \,:\, \Omega^1 \rightarrow  \bF_p^{\bP^1(\bF_p)}$.

(ii) Let $\Omega^2 = \Omega^2(\tilde{\bP}^2)$ be the $\bF_p$-vector space of $2$-forms on $\tilde{\bP}^2$, which have only simple poles and only along the set of $2(p^2+p+1)$ $\bF_p$-rational lines $\tilde{\bP}^2(\bF_p)$.
Note that $\tilde{\bP}^2(\bF_p) = \tilde{\bP}^2_{old}(\bF_p) \sqcup \tilde{\bP}^2_{new}(\bF_p)$, where $\tilde{\bP}^2_{old}(\bF_p)$ are the 'old' lines and $\tilde{\bP}^2_{new}(\bF_p)$ are the 'new' lines.
For any line $e \in \tilde{\bP}^2(\bF_p)$, define the linear map $\mbox{res}_e \,:\, \Omega^2 \rightarrow \Omega^1$, $\mbox{res}_e(\omega)$ is the residue of $\omega$ along $e$.
Define $\mbox{res}_2 =  \oplus_{e \in \tilde{\bP}^2_{old}(\bF_p)} \mbox{res}_e \,:\, \Omega^2 \rightarrow (\Omega^1)^{\tilde{\bP}^2_{old}(\bF_p)}$, and $\mbox{res} = \mbox{res}_1 \circ \mbox{res}_2  \,:\, \Omega^2 \rightarrow \bF_p^{\bP^1(\bF_p) \times \tilde{\bP}^2_{old}(\bF_p)}$.
\end{dfn}

We remark that $\dim_{\bF_p} \Omega^1= p$, $\dim_{\bF_p} \Omega^2= p^3$ and $\mbox{res}  \,:\, \Omega^2 \rightarrow \bF_p^{(p+1)(p^2+p+1)}$ is injective.

Finally, let us define the global code.

\begin{dfn} \label{dfn:HDAC-global}
Let $\Omega^2(\bar{\mX})$ be the $\bF_p$-vector space of $2$-forms on the algebraic surface $\bar{\mX}/\bF_p$ whose restrictions to each component $\bP_v$, $v$ a vertex in $X$, belong to $\Omega^2(\tilde{\bP}^2)$.
If $N$ is the number of vertices of $X$, then we shall consider $\Omega^2(\bar{\mX})$ as the following $\bF_p$-linear subspace
\begin{equation} 
\Omega^2(\bar{\mX}) \leq \oplus_v \Omega^2(\bP_v) \cong  \Omega^2(\tilde{\bP}^2)^{\oplus N} \leq \bF_p^{(p+1)(p^2+p+1)N}.
\end{equation}
\end{dfn}

We remark that for any $\omega \in \Omega^(\bar{\mX})$ and any edge $e = \{v_1,v_2\} \in X$, the following holds
\begin{equation} 
\mbox{res}_e \left( \omega|_{\bP_{v_1}} \right) = - \mbox{res}_e \left( \omega|_{\bP_{v_2}} \right).
\end{equation}
Therefore the number of constraints of the code $\Omega^2(\bar{\mX})$ is $(p+1)(p^2+p+1)N$, where $p+1$ is the number of constraints on a single edge and $(p^2+p+1)N$ is the number of edges in $X$.
This is equal to the number of degrees of freedom of $\Omega^2(\bar{\mX}) \leq \bF_p^{(p+1)(p^2+p+1)N}$.
Hence the constraints counting argument of \cite{SS} is of no use here, since it does not give any useful lower bound on the dimension of $\Omega^2(\bar{\mX})$.

Instead, by interpreting $\Omega^2(\bar{\mX})$ in terms of the cohomology of $\mX$ and arguing similarly to \cite{Mum}, but in the opposite direction (and relying on the result of Kazhdan \cite{Kaz}), one can prove that the dimension of $\Omega^2(\bar{\mX})$ grows linearly with $N$.

\begin{thm} \label{thm:HDAC-dim}
$\dim_{\bF_p} \Omega^2(\bar{\mX}) = (p-1)^2(p+1)N/3 - 1$.
\end{thm}

Unfortunately, one can show that, after possibly passing to a finite cover of $X$, the Hamming distance of the code $\Omega^2(\bar{\mX})$ is $O(N^{1/4})$.
This is due to the existence of an element in $\Omega^2(\bar{\mX})$ which is supported on a single apartment and the fact that, for principal congruence subgroups $\Gamma$, the size of the image of an apartment in $X$ is bounded by $|X|^{1/4}$. 
In particular, this implies that the distance of the code is not good.

To overcome this, one might consider supplementing the above construction with a variant of the Sipser-Spielman method \cite{SS}.
Namely, in Definition \ref{dfn:HDAC-local}, consider taking certain subspaces of $\Omega^1$, with good distance and rate.
This naturally leads us to our second code construction, which is much simpler to describe.

\subsection{High dimensional expander codes} 

Let $X = (V,E,T)$ be a finite $2$-dimensional simplicial complex.
For any $v \in V$ and $e \in E$, define the sets $E(v) =\{ e \in E \,:\, v \in e\}$, $T(v) =\{ t \in T \,:\, v \in t\}$ and $T(e) =\{ t \in T \,:\, e \subset t\}$.
Define the link of $X$ at $e$, to be the set $X_e = T(e)$.
Define the link of $X$ at $v$ to be the graph $X_v$, whose set of vertices is $E(v)$ and two vertices $e_1 , e_2 \in E(v)$ forms an edge if $e_1 \cup e_2 \in T(v)$.

Say that $X$ is $r$-regular (from the edges to triangles) if $|X_e| = r$ for any $e \in E$.
Denote $[r] = \{1,2,\ldots,r\}$. 
A labelling of $X$ is a collection of bijections $\Phi = \{\phi_e \,:\, X_e \rightarrow [r] \}_{ e \in E }$.
For any $e \in E$ and $v \in V$, denote by $\Phi_e = \{\phi_e\}$ and $\Phi_v = \{\phi_e \}_{v \in  e \in E }$, respectively.

\begin{dfn} \label{dfn:HDEC}
Let $X = (V,E,T)$ be a finite $2$-dimensional $r$-regular simplicial complex, $\Phi = \{\phi_e \,:\, X_e \rightarrow [r] \}_{ e \in E }$  a labelling of $X$, and $C_1 \leq \Sigma^{[r]}$ a code of length $r$ over the alphabet set $\Sigma$.

For any $e \in E$, define the (small) code on the link of $X$ at $e$, to be
\begin{equation}
C_e = C[X_e, \Phi_e, C_1] = \left\lbrace f \in \Sigma^{T(e)} \;:\;   f \circ \phi_e^{-1} \in C_1 \right\rbrace .
\end{equation}
For any $v \in V$, define the (intermediate) code on the link of $X$ at $v$, to be
\begin{equation}
C_v = C[X_v, \Phi_v, C_1] = \left\lbrace f \in \Sigma^{T(v)} \;:\;  \forall v \in e \in E ,\quad  f|_{T(e)} \circ \phi_e^{-1} \in C_1 \right\rbrace .
\end{equation}
Finally, define the (global) code on $X$, to be
\begin{equation}
C = C[X,\Phi,C_1] =  \left\lbrace f \in \Sigma^T \;:\;  \forall e \in E ,\quad  f|_{T(e)} \circ \phi_e^{-1} \in C_1 \right\rbrace .
\end{equation}
\end{dfn}

Let us now show how the rate and distance arguments of \cite{SS} can be extended to the code $C = C[X,\Phi,C_1]$, under some assumptions on $X$ and $C_1$.

If $\Sigma$ is a field and $C_1$ is a linear code, then the codes $C_e$ for any $e \in E$, $C_v$ for any $v \in V$ and $C$ are linear codes.
Assume $C_1$ is a linear code of rate $\rho(C_1) > \frac{2}{3}$.
By arguing as in \cite{SS} we get that the rate of the code $C = C[X,\Phi,C_1]$ (for any choice of $\Phi$), is lower bounded by
\begin{equation}
\rho(C) \geq 3 \rho(C_1) - 2.
\end{equation}

Let $X_0$ be a fixed graph and assume that $X_v \cong X_0$ for any vertex $v \in V$.
Then $C_v \cong C_0$ for any $v \in V$, where $C_0 = C[X_0, C_1]$ is the expander code of \cite{SS} w.r.t. the graph $X_0$ and the base code $C_1$.
Then arguing as in \cite{SS}, we get that the normalized distance of the code $C = C[X,\Phi,C_1]$, is lower bounded by
 \begin{equation}
 \delta(C) \geq \delta(C_0) (\delta(C_0) - \lambda(X)) \qquad \mbox{and} \qquad \delta(C_0) \geq \delta(C_1) (\delta(C_1) - \lambda(X_0)).
 \end{equation}

Hence, if $C_1$ is linear code of rate $> \frac{2}{3}$, the rate of argument of \cite{SS} can be applied to our codes, and if $X$ is a good enough high dimensional expander, then the distance argument of \cite{SS} can be applied to our codes.

\begin{que} \label{que:HDEC-1}
Are there $X$, $\Phi$ and $C_1$ satisfying the  above properties, such that $C[X,\Phi,C_1]$ is a locally testable code?
\end{que}

Let us be more concrete with our candidates of $X$.
Fix a large prime $p$ and take an infinite family of Ramanujan complexes $X$, quotients of the Bruhat-Tits building of $G=PGL_3(\bQ_p)$.  
The complex $X = (V,E,T)$ is a finite $2$-dimensional simplicial complex, such that $X_e \cong \bP^1$ for any $e \in E$, where $\bP^1$ is the projective line over $\bF_p$ (i.e. $X$ is $(p+1)$-regular), and $X_v \cong \bP^2$ for any $v \in V$, where $\bP^2$ is the graph of lines versus points of the projective plane over $\bF_p$.  
Then $\lambda(X_v) = \lambda(\bP^2) = p^{-1/2}$ for any $v\in V$, and since $X$ is Ramanujan $\lambda(X) \leq 6 p^{-1}$ (see \cite{GP}).

Given a small code $C_1$ of length $p+1$ and a labelling $\Phi = \{\phi_e \,:\, X_e \cong \bP^1 \}_{e \in E}$, the intermediate code $C_0 = C[\bP^2, \Phi, C_1]$ will be the expander code w.r.t. $\bP^2$, $\Phi$ and $C_1$.
In \cite{DDHR} the authors proved a method to propagating local testability from the intermediate code to the global code, assuming $X$ is a high dimensional expander.
A concrete case of  Question \ref{que:HDEC-1} is the following.

\begin{que} \label{que:HDEC-2}
Is there a linear code $C_1 \leq \bF_2^{\bP^1}$ of rate $> \frac{2}{3}$, normalized distance $> 2 \sqrt{3} p^{-1/2}$, and an appropriate labelling $\Phi$, such that the expander code $ C[\bP^2, \Phi, C_1]$ is a LTC?
\end{que}

One can ask a more general question.

\begin{que} \label{que:HDEC-3}
Can one define an LTC over the graph of lines points of the projective plane?
\end{que}

In  Questions \ref{que:HDEC-1}, \ref{que:HDEC-2} and \ref{que:HDEC-3}, the query complexity of the LTC is supposed to be little o of the number of bits of the code (= the number edges in the graph $\bP^2$).

We note that Definition \ref{dfn:HDEC} and Questions \ref{que:HDEC-1} and \ref{que:HDEC-2} can easily be generalized to higher dimensions.
However the above arguments of \cite{SS} and \cite{DDHR} shows that the rate, distance and local testability of the global code, whose bits are on the maximal faces of the complex, will follow from appropriate rate, distance and local testability conditions on the local codes, whose bits are on the links of the complex, by arguing inductively on the dimension of the complex.
Namely, the difficulty is to initialize this paradigm, which is precisely Question \ref{que:HDEC-2}.


\end{document}